%% file: techreport.tex
\documentclass[10pt,a4paper]{article}

\newcommand{\typeof}{0} %

\RequirePackage{ifthen}

\newcommand{\longv}[1]{\ifthenelse{\equal{\typeof}{0}}{#1}{}}
\newcommand{\shortv}[1]{\ifthenelse{\equal{\typeof}{1}}{#1}{}}
\newcommand{\longshortv}[2]{\ifthenelse{\equal{\typeof}{0}}{#1}{#2}}
\newcommand{\shortlongv}[2]{\ifthenelse{\equal{\typeof}{0}}{#2}{#1}}

\usepackage{a4wide}
\usepackage{microtype}
\usepackage{textcomp}
\usepackage{latexsym} 
\usepackage{amsfonts} 
\usepackage{amssymb}
\usepackage{amsmath}
\usepackage{amsthm}
\usepackage[colorlinks=true]{hyperref}
\usepackage[textsize=tiny]{todonotes}
\usepackage{listings}
\usepackage{boxedminipage}
\usepackage{wrapfig}
\usepackage{natbibspacing}
\usepackage{etoolbox}
\usepackage{booktabs}
\usepackage{general}
\usepackage{dice}
\usepackage[numbers,sort&compress]{natbib}

\theoremstyle{plain}
\newtheorem{proposition}{Proposition}
\newtheorem{lemma}{Lemma}
\newtheorem{theorem}{Theorem}
\newtheorem{corollary}{Corollary}

\theoremstyle{definition}
\newtheorem{definition}{Definition}
\newtheorem{example}{Example}

\theoremstyle{remark}
\newtheorem*{remark}{Remark}

\let\origthelstnumber\thelstnumber
\makeatletter

\lst@AddToHook{Init}{\def\lsthk@OnNewLine{}}
\newcommand*\showln{%
  \lst@AddToHook{OnNewLine}{%
   \let\thelstnumber\origthelstnumber%
   \advance\c@lstnumber\@ne\relax}%
}

\newcommand*\surpressln{%
  \lst@AddToHook{OnNewLine}{%
    \let\thelstnumber\relax%
    \advance\c@lstnumber-\@ne\relax%
   }%
 }%
\makeatother

\usepackage{etoolbox}
\expandafter\patchcmd\csname \string\lstinline\endcsname{%
  \leavevmode
  \bgroup
}{%
  \leavevmode
  \ifmmode\hbox\fi
  \bgroup
}{}{%
  \typeout{Patching of \string\lstinline\space failed!}%
}

\lstdefinestyle{common}{
  extendedchars=\true,
  showspaces=false,
  showstringspaces=false,
  basicstyle=\linespread{1.25}\ttfamily,
  escapechar=\%,
  emphstyle=\rmfamily\emph, 
  xleftmargin=3pt,
  xrightmargin=3pt,
  mathescape=true,
  moredelim=[is][\textcolor{darkgray}]{§}{§},
  keywordstyle=\bf,
  stringstyle=\color[rgb]{0.545098, 0.278431, 0.364706},
  commentstyle=\color[rgb]{0.698039, 0.133333, 0.133333},
  numbers=none, 
  literate={->}{$\to$}3
           {[1]}{${}_{\text{\ttfamily 1}}$}1 {[2]}{${}_{\text{\ttfamily 2}}$}1 {[n]}{${}_{\text{\ttfamily n}}$}1
           {C1}{{${\text{\ttfamily Cl}}_{\text{\ttfamily 1}}$}}2
           {C2}{{${\text{\ttfamily Cl}}_{\text{\ttfamily 2}}$}}2
           {C3}{{${\text{\ttfamily Cl}}_{\text{\ttfamily 3}}$}}2
           {comp1}{comp${}_{\text{\ttfamily 1}}$}4
           {match[walk]}{match${}_{\text{\ttfamily walk}}$}7
           {fix[walk]}{fix${}_{\text{\ttfamily walk}}$}6
           {C1u}{{${\text{\ttfamily Cl}}_{\text{\ttfamily 1}}^{\text{\ttfamily 1}}$}}2
           {C2u}{{${\text{\ttfamily Cl}}_{\text{\ttfamily 2}}^{\text{\ttfamily 1}}$}}2
           {fix[walk]u}{fix${}_{\text{\ttfamily walk}}^{\text{\ttfamily 1}}$}6
           {fix[walk]uu}{fix${}_{\text{\ttfamily walk}}^{\text{\ttfamily 2}}$}6
}
\lstdefinestyle{ocaml}{
  basicstyle=\linespread{1.1}\ttfamily,
  language=[Objective]Caml,
  morecomment=[s][\color{gray}\textit]{(*}{*)},
  style=common,
  emph={'a, a,x,f,g,xs,l,y,ys,f',g',z,cl,x'},
  otherkeywords={|,;;},
}

\lstdefinestyle{atrs}{
  style=common,
  emph={x,x',f,g,xs,xs',y,ys,l,n,z},
  numbers=none,
}

\lstdefinestyle{numbers}{
  numbers=left,
  numbersep=6pt,
  xleftmargin=12pt,
  xrightmargin=10pt,
  framexleftmargin=13pt,
  framexrightmargin=10pt,
  frame=tb,
  numberstyle=\color{gray}\scriptsize\em, 
}

\lstdefinestyle{strategy}{
  style=common,
  emph={where},
  otherkeywords={exhaustive,try},
}

\def\ml{\lstinline[style=ocaml]}
\def\atrs{\lstinline[style=atrs,mathescape]}
\def\strat{\lstinline[style=strategy,mathescape]}

\newcounter{comment}
\newcommand{\comment}[3]{\refstepcounter{comment}{\todo[color=#2]{\textbf{#1\ \thecomment.}\ #3}}}
\renewcommand{\comment}[3]{}

\renewcommand{\bibfont}{\small}

\begin{document}

\shortv{

\setlength{\pdfpageheight}{\paperheight}
\setlength{\pdfpagewidth}{\paperwidth}
\pagestyle{plain}

\conferenceinfo{ICFP'15}{August 31 -- September 2, 2015, Vancouver, BC, Canada}
\CopyrightYear{2015}
\crdata{978-1-4503-3669-7/15/08}
\doi{nnnnnnn.nnnnnnn}
}





\title{Analysing the Complexity of Functional Programs:\\ Higher-Order Meets First-Order\\ (Long Version)}\thanks{This 
work was partially supported by FWF project number J3563, FWF project number P25781-N15 and by French ANR project Elica ANR-14-CE25-0005.}

\author{Martin Avanzini\and Ugo Dal Lago\and Georg Moser}

\maketitle

\begin{abstract}
  We show how the complexity of \emph{higher-order} functional programs can
  be analysed automatically by applying program transformations to a
  defunctionalized versions of them, and feeding the result to
  existing tools for the complexity analysis of \emph{first-order}
  \emph{term rewrite systems}. This is done while carefully analysing complexity
  preservation and reflection of the employed transformations such that 
  the complexity of the obtained term rewrite system reflects 
  on the complexity of the initial program. Further, we 
  describe suitable strategies for the application of the studied
  transformations and provide ample experimental data for assessing
  the viability of our method.
\end{abstract}


\input{introduction}
\input{defunctionalization}
\input{pcf2trs}
\input{simpl}
\input{implementation}
\input{experiments}
\input{relatedwork}

\bibliography{strings-short,references}
\end{document}

%% file: introduction.tex
\section{Introduction}
\label{s:intro}

\usetikzlibrary{calc}
\usetikzlibrary{backgrounds}

\newcommand{\ibox}[5][ibox]{
  \def\esel{#4*0.25}
  \coordinate (#1) at #2 {};
  \coordinate (#1-east) at ($#2 + (0.5*#3,0)$) {};
  \coordinate (#1-west) at ($#2 + (-0.5*#3,0)$) {};
  \coordinate (#1-northwest) at ($#2 + (-0.5*#3,0.5*#4)$) {};
  \coordinate (#1-northeast) at ($#2 + (0.5*#3,0.5*#4)$) {};
  \coordinate (#1-southwest) at ($#2 + (-0.5*#3,-0.5*#4)$) {};
  \coordinate (#1-southeast) at ($#2 + (0.5*#3,-0.5*#4)$) {};

  \draw[fill=white] 
        (#1-northwest)
        -- (#1-northeast)
        -- ($(#1-southeast)+(0,\esel)$)
        -- ($(#1-southeast)+(-\esel,0)$)
        -- (#1-southwest)
        -- (#1-northwest);
  \draw[fill=black!10!white] 
        ($(#1-southeast)+(0,\esel)$)
        -- ($(#1-southeast)+(-\esel,0)$)
        -- ($(#1-southeast)+(-\esel,\esel)$)
        -- ($(#1-southeast)+(0,\esel)$);

  \node[anchor=north west,inner sep=2pt, txtlabel] (#1-text) at (#1-northwest) {\scriptsize\textsf{#5}};
}

Automatically checking programs for correctness has attracted the
attention of the computer science research community since the birth
of the discipline. Properties of interest are not necessarily
functional, however, and among the non-functional ones, noticeable
cases are bounds on the amount of resources (like time, memory and
power) programs need when executed. 

Deriving upper bounds on the resource consumption of programs is
indeed of paramount importance in many cases, but becomes undecidable
as soon as the underlying programming language is non-trivial. If the
units of measurement become concrete and close to the physical ones,
the problem gets even more complicated, given the many
transformation and optimisation layers programs are applied to before
being executed. A typical example is the one of WCET techniques
adopted in real-time systems~\cite{WEEA:ATECS:08}, which do not only need to
deal with how many machine instructions a program corresponds to, but
also with how much time each instruction costs when executed by
possibly complex architectures (including caches, pipelining, etc.), a
task which is becoming even harder with the current trend towards
multicore architectures.

As an alternative, one can analyse the \emph{abstract}
complexity of programs. As an example, one can take the number of
instructions executed by the program or the number of evaluation steps
to normal form, as a measure of its execution time. This is 
a less informative metric, which however becomes accurate if the
actual time complexity \emph{of each instruction} is kept low. One
advantage of this analysis is the independence from the specific
hardware platform executing the program at hand: the latter only needs
to be analysed once. This is indeed a path which many have followed in the
programming language community. A variety of verification techniques
have been employed in this context, like 
abstract interpretations, model checking, type systems, 
program logics, or interactive theorem provers; 
see~\cite{ABHLM:TCS:07,JHLH:POPL:10,AGM:TOCL:13,SZV:2014}
for some pointers.
If we restrict our attention to higher-order functional programs, however,
the literature becomes much sparser. 

Conceptually, when analysing the time complexity of 
higher-order programs, there is a
fundamental trade-off to be dealt with. On the one hand, one would
like to have, at least, a clear relation between the cost attributed
to a program and its actual complexity when executed: only this way
the analysis' results would be informative. On the other hand, many
choices are available as for how the complexity of higher-order programs
can be evaluated, and one would prefer one which is closer to the
programmer's intuitions. Ideally, then, one would like to work with an
informative, even if not-too-concrete, cost measure, and to be able to
evaluate programs against it fully automatically.

In recent years, several advances have been made such that the
objectives above look now more realistic than in the past, at least as
far as functional programming is concerned. First of all, some
positive, sometime unexpected, results about the invariance of unitary
cost models\footnote{In the unitary cost model, a program is
  attributed a cost equal to the number of rewrite steps needed to
  turn it to normal form.} have been proved for various forms of
rewrite systems, including the
$\lambda$-calculus~\cite{LM:LMCS:12,AD:CSL:14,AM:RTA:10}. What these
results tell us is that counting the number of evaluation steps does
\emph{not} mean underestimating the time complexity of programs, which
is shown to be bounded by a polynomial (sometime even by a linear
function~\cite{ASC:LICS:15}) in their unitary cost. This is good
news, since the number of rewrite steps is among the most intuitive
notions of cost for functional programs, at least when time is the
resource one is interested in.

But there is more. The rewriting-community has recently developed
several tools for the automated time complexity analysis of \emph{term
  rewrite system}, a formal model of computation that is at the heart
of functional programming. Examples are \aprove~\cite{GBEFFOPSSST14},
\cat~\cite{ZK:LMCS:14}, and \tct~\cite{AM:RTA:13b}. These \emph{first-order
provers} (FOPs for short) combine many different techniques, and after
some years of development, start being able to treat non-trivial
programs, as demonstrated by the result of the annual termination
competition.\footnote{\url{http://termination-portal.org/wiki/Termination_Competition}.}
This is potentially very interesting also for the complexity analysis
of \emph{higher-order functional programs}, since well-known
transformation techniques such as \emph{defunctionalisation}~\cite{Reynolds:ACM:72} 
are available, which turn higher-order
functional programs into equivalent first-order ones. This has been
done in the realm of termination~\cite{PS97,GRSST:TOPLAS:11}, but appears to be
infeasible in the context of complexity analysis. Conclusively this
program transformation approach has been reflected critical in the literature,
cf.~\cite{JHLH:POPL:10}.

\tikzstyle{txtlabel} = []
\tikzstyle{tool} = [shape=rectangle, rounded corners, draw, fill=white, minimum width=1.2cm, minimum height=1.2cm]
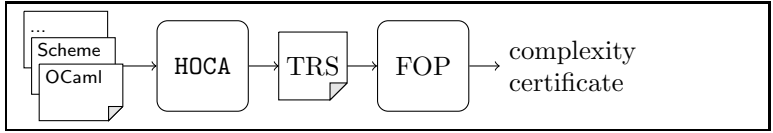
\begin{figure}
  \centering
\fbox{
\shortlongv{\begin{minipage}{.46\textwidth}}{\begin{minipage}{.65\textwidth}}
\begin{tikzpicture}
  \begin{scope}
    \ibox[C]{(-0.10,0.35)} {1.1}{0.75}{\strut{}...};
    \ibox[prog]{(0,0)}  {1.1}{0.75}{Scheme};
    \ibox[H]{(0.10,-.35)}  {1.1}{0.75}{OCaml};
  \end{scope}

  \node[tool] (hoca) at (1.7cm,0) {\hoca};
  \node[tool] (tool) at (4.6cm,0) {FOP};
  \ibox[trs]            {($(hoca)!0.5!(tool)$)}{0.9}{0.9}{};
  \node at ($(hoca)!0.5!(tool)$) {TRS};

  \node[minimum width=1.3cm,anchor=west,align=left] (answer) at ($(tool.east)+(4mm,0)$) {complexity\\certificate};

  \begin{scope}[on background layer]
    \draw[->] (prog-east) -- (hoca);
    \draw[->] (hoca) -- (trs-west);
    \draw[->] (trs-east) -- (tool);
    \draw[->] (tool) -- (answer);
  \end{scope}
\end{tikzpicture}
\end{minipage}}
\shortv{\nocaptionrule}\caption{Complexity Analysis by \hoca\ and FOPs.}\label{f:flowchart}
\end{figure}

A natural question, then, is whether time complexity analysis of
higher-order programs can indeed be performed by going through
first-order tools.
Is it possible to evaluate the unitary cost of
functional programs by translating them into first-order programs,
analysing them by existing first-order tools, and thus obtaining
meaningful and informative results? Is, e.g., plain defunctionalisation
enough?
In this paper, we show that the questions above can be answered
positively, when the necessary care is taken.
We summarise the contributions of this paper.
\begin{varenumerate}
\item We show how defunctionalisation is crucially employed in a transformation 
  from higher-order programs to first-order term rewrite
  systems, such that the time complexity of the latter reflects
  upon the time complexity of the former. More precisely, we show
  a precise correspondence between the number of reduction steps of 
  the higher-order program, and its defunctionalised version, represented
  as an \emph{applicative term rewrite systems} (see Proposition~\ref{p:pcf2trs}).

\item But defunctionalisation is not enough.
  Defunctionalised programs have a recursive structure too complicated for FOPs 
  to be effective on them. Our way to overcome this issue consists in further
  applying appropriate \emph{program transformations}. 
  These transformations must of course be proven correct to be
  viable. Moreover, we need the complexity analysis of the transformed 
  program to mean something for the starting
  program, i.e., we also prove the considered transformations to
  be at least \emph{complexity reflecting}, if not also \emph{complexity preserving}. 
  This addresses the problem that program transformations may potentially alter the
  resource usage. 
  We establish
  \emph{inlining} (see Corollary~\ref{c:narrowing}),
  \emph{instantiation} (see Theorem~\ref{t:instantiate}),
  \emph{uncurrying} (see Theorem~\ref{t:uncurrying}),
  and \emph{dead code elimination} (see Proposition~\ref{p:usablerules})
  as, at least, complexity reflecting program transformations.  

\item Still, analysing abstract program transformations is not yet
  sufficient.
  The main technical contribution of this paper
  concerns the \emph{automation} of the program transformations rather
  than the abstract study presented before.
  In particular, automating
  instantiation requires dealing with the collecting semantics of the
  program at hand, a task we pursue by exploiting tree automata and
  control-flow analysis. Moreover, we define program transformation
  strategies which allow to turn complicated defunctionalised programs
  into simpler ones that work well in practice.
\item To evaluate our approach experimentally, we have built \hoca.%
\footnote{Our tool \hoca\ is open source and available under
  \url{http://cbr.uibk.ac.at/tools/hoca/}.}  This tool is able to
translate programs written in a pure, monomorphic subset of \ocaml,
into a first-order rewrite system, written in a format
which can be understood by major first-order tools.
\end{varenumerate}
The overall flow of information is depicted in Figure~\ref{f:flowchart}.
Note that by construction, the obtained certificate \emph{reflects} onto
the runtime complexity of the initial \ocaml\ program, taking into
account the standard semantics of \ocaml.
The figure also illustrates the \emph{modularity} of the approach, as the here
studied subset of \ocaml\ just serves as a simple example language to illustrate
the method: related languages can be analysed with the same set of
tools, as long as the necessary transformation can be proven sound
and complexity reflecting. 

Our testbed includes standard higher-order functions like \lstinline|foldl| and
\lstinline|map|, but also more involved examples such as an
implementation of merge-sort using a higher-order divide-and-conquer
combinator
as well as simple parsers relying on the 
monadic parser-combinator outlined in Okasaki's functional pearl~\cite{Okasaki:JFP:98}.
We emphasise that the methods proposed here are applicable in 
the context of non-linear runtime complexities.
The obtained experimental results are quite encouraging.

\smallskip
The remainder of this paper is structures as follows. In the
next section, we present our approach abstractly on a motivating
example and clarify the challenges of our approach.
In Section~\ref{s:pcf2trs} we then present defunctionalisation
formally. Section~\ref{s:simplify} presents the \emph{transformation
pipeline}, consisting of the above mentioned program transformations. 
Implementation issues and experimental evidence is given in
Section~\ref{s:auto} and~\ref{s:experiments}, respectively. Finally,
we conclude in Section~\ref{s:relatedwork}, by discussing related work.
\shortv{An extended version of this paper with more details is
available~\cite{EV}.}


%% file: defunctionalization.tex
\section{On Defunctionalisation: Ruling the Chaos}
\label{s:defunc}
The main idea behind defunctionalisation is conceptually simple:
function-abstractions are represented as first-order values; calls to
abstractions are replaced by calls to a globally defined
\emph{apply-function}.  Consider for instance the following
\ocaml-program:
\begin{lstlisting}[style=ocaml]
let comp f g = fun z->f (g z) ;;
let rec walk xs = 
  match xs with 
    [] -> (fun z->z)
  | x::ys -> comp (walk ys) 
               (fun z->x::z) ;;
let rev l = walk l [] ;;	  
let main l = rev l ;;
\end{lstlisting}
Run on a list of $n$ elements, \ml|walk| first constructs
a function which reverses its first argument and appends it to the
second argument.  This function, which can be easily defined by
recursion, is fed in \ml|rev| with the empty list. 
The function \ml|main| only serves the purpose
of indicating the complexity of \emph{which} function we are interested
at. 

Defunctionalisation can be understood already at this level. 
We first define a datatype for representing the three
abstractions occurring in the program:
\begin{lstlisting}[style=ocaml]
type 'a cl = 
  C1 of 'a cl * 'a cl (*$\,\,$fun z->f (g z)$\,\,$*)
| C2                  (*$\,\,$fun z->z$\,\,$*)
| C3 of 'a            (*$\,\,$fun z->x::z$\,\,$*)
\end{lstlisting}
More precisely, an expression of type \ml|'a cl| represents a function
\emph{closure}, whose arguments are used to store assignments to free
variables. An infix operator (\ml|@|), modelling application, can then be
defined as follows:%
\footnote{The definition is rejected by the \ocaml\ type-checker,
  which however, is not an issue in our context.}
\begin{lstlisting}[style=ocaml]
let rec (@) cl z = 
  match cl with 
    C1(f,g) -> f @ (g @ z)
  | C2 -> z 
  | C3(x) -> x::z ;;
\end{lstlisting}
Using this function, we 
arrive at a first-order version of the original higher-order function:
\begin{lstlisting}[style=ocaml]
let comp f g = C1(f,g) ;;
let rec walk xs = 
  match xs with 
    [] -> C2
  | x::ys -> comp (walk ys) C3(x) ;;
let rev l = walk l @ [] ;;
let main l = rev l ;;
\end{lstlisting}
Observe that now the recursive function \ml|walk| constructs an
explicit representation of the closure computed by its original
definition. The function (\ml|@|) carries out the remaining
evaluation.  
This program can now already be understood as a first-order rewrite system.

Of course, a systematic construction of the defunctionalized program requires some care. 
For instance, one has to deal with closures that originate from partial function applications. 
Still, the construction is quite easy to mechanize,  
see Section~\ref{s:pcf2trs} for a formal treatment.
On our running example, this program transformation results in the rewrite system $\atrsrev$, which looks as follows:%
\footnote{
In $\atrsrev$, rule~\eqref{atrsrev:fix}
reflects that, under the hood, we treat recursive let expressions as syntactic sugar 
for a dedicated fixpoint operator.}

\begin{lstlisting}[style=atrs,style=numbers]
C1(f,g) @ z -> f @ (g @ z)
C2 @ z -> z
C3(x) @ z -> x::z
comp1(f) @ g -> C1(f,g)
comp @ f -> comp1(f)
match[walk]([]) -> C2
match[walk](x::ys) -> %\label{atrsrev:walk_cond}\surpressln%
  comp @ (fix[walk] @ ys) @ C3(x) %\showln%
walk @ xs -> match[walk](xs)     %\label{atrsrev:walk}%
fix[walk] @ xs -> walk @ xs      %\label{atrsrev:fix}%
rev @ l -> fix[walk] @ l @ []
main(l) -> rev @ l
\end{lstlisting}

Despite its conceptual simplicity, current FOPs are unable to effectively 
analyse \emph{applicative} rewrite systems, 
such as the one above.
The reason this happens lies in the way FOPs work, which itself reflects the
state of the art on formal methods for complexity analysis of first-order rewrite systems. 
In order to achieve composability of the analysis, the given system is
typically split into smaller parts (see for example~\cite{AM:IC:15}),
and each of them is analysed separately. Furthermore,
contextualisation (aka \emph{path analysis}~\cite{HM:LPAR:08}) and a
suitable form of flow graph analysis (or \emph{dependency pair analysis}~\cite{HM:IJCAR:08,NEG:CADE:11}) is performed.  However, at
the end of the day, syntactic and semantic basic techniques, like path
orders or interpretations~\cite[Chapter~6]{Terese} are employed. 
All these
methods focus on the analysis of the given defined symbols (like for
instance the application symbol in the example above) and fail if their recursive
definition is too complicated. Naturally this calls for a special
treatment of the applicative structure of the
system~\cite{HMZ:JAR:13}.

How could we get rid of those (\atrs|@|), thus highlighting the deep
recursive structure of the program above? Let us, for example,
focus on the rewriting rule 
\begin{lstlisting}[style=atrs]
C1(f,g) @ z -> f @ (g @ z) , 
\end{lstlisting}
which is particularly nasty for FOPs, given that the variables \atrs|f|
and \atrs|g| will be substituted by unknown functions, which could potentially
have a very high complexity. How could we \emph{simplify} all this?
The key observation is that although this rule 
tells us how to compose two \emph{arbitrary} closures, 
only very few instances of the rule above are needed, 
namely those were \atrs|g| is of the form \atrs|C3(x)|, and
\atrs|f| is either \atrs|C2| or again of the form \atrs|C1(f',g')|.
This crucial information can be retrieved in the 
so-called \emph{collecting semantics}~\cite{NNH:2005}
of the term rewrite system above, 
which precisely tells us which object will possibly be substituted for rule
variables along the evaluation of certain families of terms.  Dealing
with all this fully automatically is of course impossible, but
techniques based on tree automata, and inspired by those in~\cite{Jones:TCS:07}
can indeed be of help.

Another useful observation is the following: function symbols 
like, e.g., \atrs|comp| or \atrs|match[walk]| are essentially useless:
their only purpose is to build intermediate closures, or to control program flow:
One could simply shortcircuit them, using a form of \emph{inlining}. And
after this is done, some of the left rules are \emph{dead code},
and can thus be  eliminated from the program. 
At the end of the day, we arrive at a truly first-order system and 
\emph{uncurrying} brings it to a format most suitable for FOPs.

If we carefully apply the just described ideas to the example above,
we end up with the following first-order system, called $\trsrev$, which is precisely what
\hoca\ produces in output:
\begin{lstlisting}[style=atrs,style=numbers]
C1u(C2,C3(x),z) -> x::z %\label{uc:C1:1}%
C1u(C1(f,g),C3(x),z) -> C1u(f,g,x::z) %\label{uc:C2:2}%
fix[walk]u([]) -> C2   %\label{uc:walk:1}%
fix[walk]u(x:ys) -> C1(fix[walk]u(ys),C3(x))%\label{uc:walk:2}%
main(l) -> []
main(x:ys) -> C1u(fix[walk]u(ys),C3(x),[])
\end{lstlisting}
This term rewrite system is equivalent to $\atrsrev$ from above, both extensionally and in
terms of the underlying complexity. However, the FOPs we have
considered can indeed conclude that \atrs|main| has linear
complexity, a result that can be easily lifted back to the original
program.

Sections~\ref{s:simplify} and~\ref{s:auto} are concerned with a
precise analysis of the program transformations we employed when
turning $\atrsrev$ into $\trsrev$. 
Before that, we recap central definitions
in the next section.


%% file: pcf2trs.tex
\section{Preliminaries}\label{s:pcf2trs}
The purpose of this section is to give some preliminary notions about
the $\lambda$-calculus, term rewrite systems, and translations between
them; see~\cite{BN:1998,Terese,Pierce:2002} for further reading.

To model a reasonable rich but pure and monomorphic functional language,
we consider a typed $\lambda$-calculus with constants and fixpoints
akin to Plotkin's \PCF~\cite{Plotkin:TCS:77}. To seamlessly express
programs over algebraic datatypes, we allow constructors and pattern
matching.  To this end, let $\seq[k]{\conone}$ be finitely many
constructors, each equipped with a fixed \emph{arity}. 
The syntax of \PCF-programs is given by the following grammar:
\begin{align*}
  \EXP\quad\expone,\exptwo~&\bnfdef~ 
  \varone 
  \mid \conone_i(\vec{\expone}) 
  \mid \lam{\varone}{\expone}
  \mid \app{\expone}{\exptwo}
  \mid \fix{\varone}{\expone}\\
  &\mid \case{\expone}{\conone_1(\vec{\varone}_1) \mapsto \expone_1; \cdots ; \conone_k(\vec{\varone}_k) \mapsto \expone_k} \tkom
\end{align*}
where $\varone$ ranges over variables. Note that the variables
$\vec{\varone}_i$ in a match-expression are considered bound in
$\expone_i$. A simple type system can be easily defined 
based on a single ground type, and on the usual arrow type
constructor. We claim that extending the language with products and
coproducts would not be problematic.

We adopt \emph{weak call-by-value} semantics, the definition is
standard, see e.g.~\cite{Harper:PFP:2012}.  Here \emph{weak} means
that reduction under any $\lambda$-abstraction
$\lam{\varone}{\expone}$ and any fixpoint-expressions
$\fix{\varone}{\expone}$ is prohibited.  \emph{Call-by-value} means
that in a redex $\app{\expone}{\exptwo}$, the expression $\expone$ has
to be evaluated to a value first.  A match-expression
$\case{\expone}{\conone_1(\vec{\varone}_1) \mapsto \expone_1; \cdots ;
  \conone_k(\vec{\varone}_k) \mapsto \expone_k}$ is evaluated by first
evaluating the guard $\expone$ to a value $\conone_i(\vec{\valone})$,
reduction then continues with the corresponding case-expression
$\expone_i$ with values $\vec{\valone}_i$ substituted for variables
$\vec{\varone}_i$. The one-step weak call-by-value reduction relation
is denoted by $\pcf$.  Elements of the term algebra over constructors
$\seq[k]{\conone}$ embedded in our language are collected in $\DATA$.
A \emph{$\PCF$-program} with $n$ \emph{input arguments} is a closed
expression $\progone = \lambda \varone_1 \cdots \lambda
\varone_n. \expone$ of first-order type. What this implicitly means is
that we are interested in an analysis of programs with a possibly very
intricate internal higher-order structure, but whose arguments are
values of ground type.  
This is akin to the setting in~\cite{BL:CSL:12} and provides an intuitive 
notion of runtime complexity for higher-order program, without having to
rely on ad-hoc restrictions on the use of function-abstracts 
(as e.g.~\cite{JHLH:POPL:10}).
This way we also ensure that the abstractions
reduced in a run of $\progone$ are the ones found in $\progone$, an
essential property for performing defunctionalisation.  We assume that
variables in $\progone$ have been renamed apart, and we impose a total
order on variables in $\progone$. The free variables $\FV(\expone)$ in
the body $\expone$ of $\progone$ can this way be defined as an ordered
sequence of variables.
\begin{example}\label{ex:pcf}
We fix constructors \atrs|[]| and (\atrs|::|) for lists, the latter we write 
infix. 
Then the program computing the reverse of a function, as described
in the previous section, can be seen as the \PCF\ term
$\progone_{\mathsf{rev}} \defsym \lam{l}{\app{\mathit{rev}}{l}}$ where 
\begin{align*}
  \mathit{rev} & = \lam{l}{\app{\app{\fix{w}{\mathit{walk}}}{l}}{{\atrs|[]|}}} \tspkt\\
  \mathit{walk} & = \lambda xs. \CASE\ xs\ 
         \left\{ 
         \begin{array}{@{\,}r@{}l@{\ \ }}
           \!{\atrs|[]|} &{} \mapsto \lam{z}{z}\tspkt \\
           \!{\atrs|x::ys|} &{} \mapsto \app{\app{\mathit{comp}}{(\app{w}{ys})}}{(\lam{z}{{\atrs|x::z|}})}
         \end{array}\right\}; \\
  \mathit{comp} & = \lam{f}{\lam{g}{\lam{z}{\app{f}{(\app{g}{z})}}}} \tpkt
\end{align*}
\end{example}

The second kind of programming formalism we will deal with is the
one of \emph{term rewrite systems} (TRSs for short).
Let $\FS = \{ \seq[n]{\funone} \}$ be a set of function symbols, 
each equipped again with an arity, 
the \emph{signature}.
We denote by $\termtwo,\termone,\dots$ terms over the signature $\FS$, 
possibly including variables. 
A \emph{position} $\posone$ in $\termone$ is a finite sequence of integers, such 
that the following definition of \emph{subterm at position} $\posone$ is well-defined:
$\subtermAt{\termone}{\posempty} = \termone$ for the \emph{empty position} $\posempty$, 
and $\subtermAt{\termone}{ip} = \subtermAt{\termone_i}{p}$ for $\termone=\funone(\seq[k]{\termone})$. 
For a position $\posone$ in $\termone$, we denote by $\ctx[\posone]{\termone}[\termtwo]$ the term 
obtained by replacing the subterm at position $\posone$ in $\termone$ by the term $\termtwo$. 
A \emph{context} $\ctxone$ is a term containing one occurrence of a special symbol $\hole$, the \emph{hole}. 
We define $\ctx{\ctxone}[\termone] \defsym \ctx[\posone]{\ctxone}[\termone]$ for $\posone$ 
the position of $\hole$ in $\ctxone$, i.e., $\subtermAt{\ctxone}{\posone} = \hole$. 

A \emph{substitution}, is a finite mapping $\sigma$ from variables to terms. 
By $\termone\sigma$ we denote the term obtained by replacing in $\termone$ 
all variables $\varone$ in the domain of $\sigma$ by $\sigma(\varone)$. 
A substitution $\sigma$ is \emph{at least as general} as a substitution $\tau$ if 
there exists a substitution $\tau'$ such that $\tau(\varone) = \sigma(\varone)\tau'$ for each 
variable $\varone$. 
A term $\termone$ is an instance of a term $\termtwo$ if there exists a substitution $\sigma$, 
with $\termtwo\sigma = \termone$; the terms $\termone$ and $\termtwo$ \emph{unify} if 
there exists a substitution $\mguone$, the \emph{unifier}, such that $\termone\mguone = \termtwo\mguone$. 
If two terms are unifiable, then there exists a \emph{most general unifier} (mgu for short).

A \emph{term rewrite system} $\TRSone$
is a finite set of rewrite rules, i.e., directed equations $\funone(\seq[k]{l}) \to r$ 
such that all variables occurring in the \emph{right-hand side} $r$ occur also 
in the \emph{left-hand side} $\funone(\seq[k]{l})$. The roots of left-hand sides, 
the \emph{defined symbols} of $\TRSone$, are collected in $\DS_\TRSone$, the remaining symbols $\FS \setminus \DS_\TRSone$
are the \emph{constructors} of $\TRSone$ and collected in $\CS_\TRSone$. 
Terms over the constructors $\CS_\TRSone$ are considered \emph{values} and collected in $\Val[\CS_\TRSone]$. 
\shortlongv{%
We denote by $\rew[\TRSone]$ the one-step rewrite relation of $\TRSone$, imposing \emph{call-by-value} 
semantics. 
Call-by-value means that variables are assigned elements of $\Val[\CS_\TRSone]$.}
{%
We adopt \emph{call-by-value} semantics for TRSs, see Figure~\ref{fig:trssemantics} where
the \emph{call-by-value rewrite relation} $\rew[\TRSone]$ is defined. 
\begin{figure}
  \centering
  \fbox{
    \begin{minipage}{.95\textwidth}
      \begin{equation*}
        \frac{\funone(l_1,\dots,l_n) \to r \in \TRSone \qquad \text{$x\sigma \in \Val[\CS_\TRSone]$ for all variables $x$ occurring in $\funone(l_1,\dots,l_n)$}}
        {\funone(l_1\sigma,\dots,l_n\sigma) \rew[\TRSone] r\sigma}
      \end{equation*}
      \begin{equation*}
        \frac{\termtwo_i \rew[\TRSone] \termone_i}
        {\funone(\termtwo_1,\dots,\termtwo_i,\dots\termtwo_n) 
          \rew[\TRSone] 
          \funone(\termtwo_1,\dots,\termone_i,\dots\termtwo_n)}
      \end{equation*}
      \vspace{2mm}
    \end{minipage}
  }
  \caption{Call-by-value rewrite relation with respect to a TRS $\TRSone$.}\label{fig:trssemantics}
\end{figure}}

Throughout the following, we consider \emph{non-ambiguous} rewrite
systems, that is, the left-hand sides are pairwise
\emph{non-overlapping}.  Even thought $\rew[\TRSone]$ may be
non-deterministic, the following special case of the parallel moves
lemma~\cite{BN:1998} tells us that this form of non-determinism is not
harmful for complexity-analysis.
\begin{proposition}\label{p:nonambiguous}
    For a non-ambiguous TRS $\TRSone$, all \emph{normalising reductions} of $\termone$ have the same length, i.e,
    if $\termone \rsl[\TRSone]{m} \termthree_1$ and $\termone \rsl[\TRSone]{n} \termthree_2$ for 
    two irreducible terms $\termthree_1$ and $\termthree_2$, 
    then $\termthree_1 = \termthree_2$ and $m = n$.
\end{proposition}

An \emph{applicative term rewrite system} (\emph{ATRS} for short) is usually defined as a TRS over a signature
consisting of a finite set of nullary function symbols and one dedicated binary symbol (${\api}$), the \emph{application symbol}.
We follow the usual convention that (${\api}$) associates to the left.
Here, we are more liberal and just assume the presence of ($\api$), 
and allow function symbols that take more than one argument.
Throughout the following, we are foremost dealing with ATRSs, which we denote by $\atrsone,\atrstwo$ below. 
We also write ($\api$) infix and assume that it associates to the left.

In the following, we show that every $\PCF$-program $\progone$ can be
seen as an applicative term rewrite system $\atrsp$. 
To this end, we first define an \emph{infinite schema} $\atrsPCF$ of rewrite rules which 
allows us to evaluate the whole of $\PCF$. 
The signature underlying $\atrsPCF$ contains, besides 
the application-symbol ($\api$) and constructors $\seq[k]{\conone}$, 
the following function symbols, called \emph{closure constructors}:
(i)~for each \PCF\ term $\lam{\varone}{\expone}$ with $n$ free variables an
$n$-ary symbol $\clam{\varone}{\expone}$; 
(ii)~for each \PCF\ term $\fix{\varone}{\expone}$ with $n$ free variables an
$n$-ary symbol $\cfix{\varone}{\expone}$; and
(iii)~for each match-expression $\case{\expone}{cs}$ with $n$ free
variables a symbol $\ccase{cs}$ of arity $n+1$.
Furthermore, We define a mapping $\tr{\cdot}$ from \PCF\
terms to $\atrsPCF$ terms as follows. 
\begin{align*}
  \tr{\varone} & \defsym \varone \tspkt\\ 
  \tr{\lam{\varone}{\expone}} & \defsym \clam{\varone}{\expone}(\vec{\varone})\text{, where $\vec{\varone} = \FV(\lam{\varone}{\expone})$}\tspkt\\
  \tr{\conone_i(\seq[k]{\expone})} & \defsym \conone_i(\mapm\tr[k]{\expone}) \tspkt\\
  \tr{\app{\exptwo}{\expone}} & \defsym \apl{\tr{\exptwo}}{\tr{\expone}}\tspkt\\ 
  \tr{\fix{\varone}{\expone}} & \defsym \cfix{\varone}{\expone}(\vec{\varone})\text{, where $\vec{\varone} = \FV(\fix{\varone}{\expone})$}\tspkt\\ 
  \tr{\case{\expone}{cs}} & \defsym \ccase{cs}(\tr{\expone},\vec{\varone})\text{, where $\vec{\varone} = \FV(\{cs\})$}\tpkt
\end{align*}
Based on this interpretation, each closure constructor is equipped with one 
or more of the following \emph{defining rules}:
\begin{align*}
\apl{\clam{\varone}{\expone}(\vec{\varone})}{\varone}&\to\tr{\expone} \tspkt\\
\apl{\cfix{\varone}{\expone}(\vec{\varone})}{\vartwo}&\to \apl{\tr{\subst{\expone}{\varone}{\fix{\varone}{\expone}}}}{\vartwo}\text{ , where $\vartwo$ is fresh;}\\
\ccase{cs}(\conone_i(\vec{\varone}_i),\vec{\varone})&\to \tr{\expone_i}\text{ , for $i = 1,\dots,k$.}
\end{align*}
Here, we suppose $cs = \{\conone_1(\vec{\varone}_1) \mapsto \expone_1; \cdots ; \conone_k(\vec{\varone}_k) \mapsto \expone_k\}$.

For a program $\progone = \lambda \varone_1 \cdots \lambda
\varone_n. \expone$, the ATRS $\atrsp$ 
(i)~contains a rule $\fmain(\seq[n]{\varone}) \to \tr{\expone}$, where $\fmain$ is a dedicated function symbol; 
together with
(ii) the least subset of $\atrsPCF$ that defines all closure constructors occurring in $\atrsp$. 
Crucial, $\atrsp$ is always finite, in fact, the size of $\atrsp$ 
is linearly bounded in the size of $\progone$\shortlongv{~\cite{EV}}{, see below}. 
\begin{remark}
This statement becomes trivial if we consider alternative defining rule 
\[
\apl{\cfix{\varone}{\expone}(\vec{\varone})}{\vartwo} \to \apl{\subst{\tr{\expone}}{\cfix{\varone}{\expone}(\vec{\varone})}{\varone}}{\vartwo} \tkom
\]
which would also correctly model the semantics of fixpoints $\fix{\varone}{\expone}$.
Then the closure constructors occurring in $\atrsp$ are all obtained from sub-expressions of $\progone$.
Our choice is motivated by the fact that closure constructors
of fixpoints are propagates to call sites, 
something that facilitates our transformation approach to complexity analysis.
\end{remark}

\begin{example}\label{ex:tr}
  The expression $\progone_{\mathsf{rev}}$ from Example \ref{ex:pcf}
  gets translated into the ATRS $\atrsp[\progone_{\mathsf{rev}}]=\atrsone_{\mathsf{rev}}$
  we introduced in Section~\ref{s:defunc}.
\end{example}

\shortlongv{%
We obtain the following simulation result, a proof of which
can be found in~\cite{EV}, but also in~\cite{LM:LMCS:12}.}{%
We obtain the following simulation result}
\begin{proposition}\label{p:pcf2trs}
  Every $\pcf$-reduction of an expression $\progone\ \dataone_1\ \cdots\ \dataone_n$ ($\dataone_j \in \DATA$) is simulated step-wise 
  by a call-by-value $\atrsp$-derivation starting from $\fmain(\seq[n]{\dataone})$. 
\end{proposition}

As the inverse direction of this
proposition can also be stated, $\atrsp$ can be seen as a sound and
complete, in particular step-preserving, implementation of the
$\PCF$-program $\progone$.

In correspondence to Proposition~\ref{p:pcf2trs}, we define the runtime complexity of an ATRS $\atrsone$
as follows. 
As above, only terms $\dataone \in \DATA$ built from the constructors $\CS$ are considered valid inputs. 
The \emph{runtime of $\atrsone$ on inputs} $\seq[n]{\dataone}$ is defined as the length of the longest rewrite sequence 
starting from $\fmain(\seq[n]{\dataone})$. The \emph{runtime complexity function} 
is defined as the (partial) function which maps the natural number $m$
to the maximum runtime of $\atrsone$ on inputs $\seq[n]{\dataone}$ with $\sum_{j=1}^n \size{\dataone_j} \leqslant m$, 
where the \emph{size} $\size{\dataone}$ is defined as the number of occurrences of constructors in $\dataone$.

Crucial, our notion of runtime complexity corresponds to the notion
employed in first-order rewriting and in particular in FOPs. Our
simple form of defunctionalisation thus paves the way to our primary
goal: a successful complexity analysis of $\atrsp$ with
rewriting-based tools can be relayed back to the $\PCF$-program
$\progone$.


%% file: simpl.tex
\section{Complexity Reflecting Transformations}\label{s:simplify}
The result offered by Proposition~\ref{p:pcf2trs} is remarkable, but
is a Pyrrhic victory towards our final goal: as discussed in
Section~\ref{s:defunc}, the complexity of defunctionalised programs is
hard to analyse, at least if one wants to go via FOPs. It is then time
to introduce the four program transformations that form our toolbox,
and that will allow us to turn defunctionalised programs into ATRSs
which are easier to analyse.

In this section, we describe the four transformations abstractly,
without caring too much about \emph{how} one could implement them.
Rather, we focus on their correctness and, even more importantly for
us, we verify that the complexity of the transformed program is not
too small compared to the complexity of the original one. We will also
show, through examples, how all this can indeed be seen as a way to
simplify the recursive structure of the programs at hand.

A \emph{transformation} is a partial function $\transone$ from ATRSs
to ATRSs.  In the case that $\transone(\atrsone)$ is undefined, the
transformation is called \emph{inapplicable} to $\atrsone$.  We call
the transformation $\transone$ (\emph{asymptotically})
\emph{complexity reflecting}
if for every ATRS $\atrsone$, the runtime
complexity of $\atrsone$ is bounded (asymptotically) by the runtime
complexity of $\transone(\atrsone)$, whenever $\transone$ is
applicable on $\atrsone$.  Conversely, we call $\transone$
(\emph{asymptotically}) \emph{complexity preserving} if the runtime complexity of
$\transone(\atrsone)$ is bounded (asymptotically) by the complexity of $\atrsone$,
whenever $\transone$ is applicable on $\atrsone$.
The former condition states a form of \emph{soundness:} if $\transone$
is complexity reflecting, then a bound on the runtime complexity
of $\transone(\atrsone)$
can be relayed back to $\atrsone$. The latter conditions states
a form of \emph{completeness}: application of a complexity preserving transformation
$\transone$ will not render our analysis ineffective, simply because $\transone$ translated
$\atrsone$ to an inefficient version.
We remark that the
set of complexity preserving (complexity reflecting)
transformations is closed under composition.

\subsection{Inlining}\label{s:inlining}
Our first transformation constitutes a form of \emph{inlining}. This
allows for the elimination of auxiliary functions, this way making the
recursive structure of the considered program apparent.

Consider the ATRS $\atrsrev$ from Section~\ref{s:defunc}.
There, for instance, the call to \atrs|walk| in the definition of
\atrs|fix[walk]| could be \emph{inlined}, thus resulting in a new
definition:
\begin{lstlisting}[style=atrs]
fix[walk] @ xs -> match[walk](xs)
\end{lstlisting}
Informally, thus, inlining consists in modifying the right-hand-sides
of ATRS rules by rewriting subterms, according to the ATRS
itself. 
We will also go beyond rewriting, 
by first specializing arguments sufficiently so
that a rewrite triggers. In the above rule for instance, 
\atrs|match[walk]| cannot be inlined immediately, simply 
because \atrs|match[walk]| is defined itself by case analysis on \atrs|xs|. 
To allow inlining of this function nevertheless, we specialize \atrs|xs| to the patterns 
\atrs|[]| and \atrs|x::ys|, the patterns underlying the case analysis of \atrs|match[walk]|.
This results in two alternative rules for \atrs|fix[walk]|, namely
\begin{lstlisting}[style=atrs]
fix[walk] @ [] -> match[walk]([])
fix[walk] @ (x::ys) -> match[walk](x::ys) .
\end{lstlisting}
Now we can inline \atrs|match[walk]|, and as a consequence the 
rules defining \atrs|fix[walk]| are easily seen to be structurally recursive,
a fact that FOPs can recognise and exploit.

A convenient way to formalise inlining is by way of
\emph{narrowing}~\cite{BN:1998}. 
We say that a term $\termtwo$
\emph{narrows} to a term $\termone$ at a non-variable position
$\posone$ in $\termtwo$, in notation $\termtwo \narrow{\mguone}{\posone} \termone$, 
if there exists a rule $l \to r \in \atrsone$ 
such that $\mguone$ is a unifier of left-hand side $l$
and the subterm $\subtermAt{\termtwo}{\posone}$ (after renaming apart variables in $l \to r$ and $\termtwo$) 
and $\termone = \ctx[\posone]{\termtwo\mguone}[r\mguone]$.
In other words, the instance $\termtwo\mguone$ of $\termtwo$ rewrites to $\termone$ at position $\posone$
with rule $l \to r \in \atrsone$. 
The substitution $\mguone$ is just enough to uncover the corresponding redex in $\termtwo$.
Note however that the performed rewrite step is not necessarily call-by-value, the mgu $\mguone$ 
could indeed contain function calls.
We define the set of all \emph{inlinings} of a rule $l \to r$ at position $\posone$ which is labeled by a defined symbol by
\[
\narrowings{\posone}{l \to r} \defsym \{ l \mguone \to r' \mid r \narrow{\mguone}{\posone} r' \} \tpkt
\]
The following example demonstrates inlining through narrowing.
\begin{example}
Consider the substitutions $\mguone_1 = \{{\atrs|xs|}\mapsto \text{\atrs|[]|}\}$ and
$\mguone_2 = \{{\atrs|xs|}\mapsto {\atrs|x::ys|}\}$. Then we have
\begin{align*}
{\atrs|match[walk](xs)|}
& \narrow[\atrsrev]{\mguone_1}{\posempty}
\text{\atrs|C2|}\\
\text{\atrs|match[walk](xs)|}
& \narrow[\atrsrev]{\mguone_2}{\posempty}{\atrs|comp @ (fix[walk] @ ys) @ C3(x)|}\tpkt
\end{align*}
Since no other rule of $\atrsrev$ unifies with the right-hand side 
\atrs|match[walk](xs)|, the set 
\[
\narrowings[\atrsrev]{\posempty}{\text{\atrs|fix[walk] @ xs -> match[walk](xs)|}}
\]
consists of the two rules 
\begin{lstlisting}[style=atrs]
fix[walk] @ [] -> C2
fix[walk] @ (x::ys) ->
  comp @ (fix[walk] @ ys) @ C3(x) .
\end{lstlisting}
\end{example}

Inlining is in general not complexity reflecting.
Indeed, inlining is employed by many compilers as a program optimization technique. 
The following examples highlight two issues we have to address. The first example
indicates the obvious: in a call-by-value setting,
inlining is not asymptotically complexity reflecting, if potentially expensive function calls
in arguments are deleted. 
\begin{example}
  Consider the following inefficient system:
\begin{lstlisting}[style=atrs,style=numbers]
k(x,y) -> x
main(0) -> 0
main(S(n)) -> k(main(n),main(n)) %\label{r:narrowing:1}%
\end{lstlisting}
Inlining \atrs|k| in the definition of \atrs|main| results in an alternative definition
$\atrs|main(S(n)) -> main(n)|$
of rule~\eqref{r:narrowing:1}, eliminating one of the two recursive calls and 
thereby reducing the complexity from exponential to linear. 
\end{example}

The example motivates the following, easily decidable, condition. 
Let $l \to r$ denote a rule whose right-hand side is subject to inlining at position $\posone$. 
Suppose the rule $u \to v \in \atrsone$ is unifiable with the subterm $\subtermAt{r}{\posone}$ of the right-hand side $r$, 
and let $\mguone$ denote the most general unifier. 
Then we say that inlining $\subtermAt{r}{\posone}$ with $u \to v$ 
is \emph{redex preserving} if whenever $x\mguone$ contains a defined symbol of $\atrsone$, 
then the variable $x$ occurs also in the right-hand side $v$. 
The inlining $l \to r$ at position $\posone$ is called \emph{redex preserving} if inlining $\subtermAt{r}{p}$ 
is redex preserving with \emph{all rule} $u \to v$ that unify with $\subtermAt{r}{\posone}$.
Redex-preservation thus ensures that inlining does not delete potential function calls, apart 
from the inlined one.
In the example above, inlining \atrs|k(main(n),main(n))| is not redex preserving because
the variable \atrs|y| is mapped to \atrs|main(n)| by the underlying unifier, 
but \atrs|y| is deleted in the inlining rule \atrs|k(x,y) -> x|.

Our second example is more subtle and arises when the studied rewrite
system is under-specified:
\begin{example}
  Consider the system consisting of the following rules.
\begin{lstlisting}[style=atrs,style=numbers]
h(x,0) -> x
main(0) -> 0
main(S(n)) -> h(main(n),n) %\label{r:narrowing:2}%
\end{lstlisting}
Inlining \atrs|h| in the definition of \atrs|main| will specialise the
variable \atrs|n| to \atrs|0| and thus replaces rule~\eqref{r:narrowing:2}
by \atrs|main(S(0)) -> main(0)|. Note that the runtime complexity
of the former system is linear, whereas its runtime complexity is constant
after transformation.
\end{example}
Crucial for the example, the symbol \atrs|h| is not \emph{sufficiently defined}, i.e., 
the computation gets stuck after completely unfolding \atrs|main|.
To overcome this issue, we require that inlined functions are sufficiently defined. 
Here a defined function symbol $\funone$ is called \emph{sufficiently defined}, with respect to an ATRS $\atrsone$,
if all subterms $\funone(\vec{\termone})$ occurring in a reduction
of $\fmain(\seq[n]{\dataone})$ ($\dataone_j \in \DATA$) are reducible.
This property is not decidable in general. 
Still, the ATRSs obtained from the translation in Section~\ref{s:pcf2trs}
satisfy this condition for all defined symbols: by construction, 
reductions do not get stuck. 
Inlining, and the transformations discussed below, 
preserve this property.

We will now show that under the above outlined conditions, inlining is indeed complexity reflecting.
Fix an ATRS $\atrsone$.
\longv{%
The following auxiliary lemma follows by a standard induction on the length
of derivations, see e.g.~\cite{HMZ:JAR:13}. As a consequence, we can assume
that reductions have a very specific form.
\begin{lemma}\label{l:narrow:irew}
  \begin{enumerateenv}
  \item\label{l:narrow:irew:ctx}
    If $\ctxone[\termone] \rsl[\atrsone]{m} \termthree$ is a normalizing derivation, then
    $\ctxone[\termone] \rsl[\atrsone]{m_1} \ctxone[\termtwo] \rsl[\atrsone]{m_2} \termthree$
    for some normalform $\termtwo$ of $\termone$ and $m_1,m_2\in\N$  with $m_1 + m_2 = m$.
  \item\label{l:narrow:irew:subst}
    If $\termone\sigma \rsl[\atrsone]{m} \termthree$ is a normalizing derivation, then
    $\termone\sigma \rsl[\atrsone]{m_1} \termone\tau \rsl[\TRSone]{m_2} \termthree$
    for some normalized substitutions $\tau$ and $m_1,m_2\in\N$  with $m_1 + m_2 = m$.
  \end{enumerateenv}
\end{lemma}

}
In proofs below, we denote by $\frew[\atrsone]$ an extension of $\rew[\atrsone]$ where
not all arguments are necessarily reduced, but where still a step cannot delete redexes:
$\termtwo \frew[\atrsone] \termone$ if $\termtwo = C[l\sigma]$ and
$\termone = C[r\sigma]$ for a context $C$, rule $l \to r \in \atrsone$
and a substitution $\sigma$ which satisfies $\sigma(\varone) \in \Val[\CS_\atrsone]$ for all variables 
$\varone$ which occur in $l$ but not in $r$. 
By definition, ${\rew[\atrsone]} \subseteq {\frew[\atrsone]}$.
The relation $\frew[\atrsone]$ is just enought to capture rewrites performed on right-hand sides 
in a complexity reflecting inlining. 


The next lemma collects the central points of our correctness proof.
Here, we first considers the effect of replacing a single application of a rule $l \to r$ 
with an application of a corresponding rule in $\narrowings{\posone}{l \to r}$. 
As the lemma shows, this is indeed always possible, provided the inlined function 
is sufficiently defined. Crucial, inlining preserves not only semantics, 
but complexity reflecting inlining does not optimize the ATRS under consideration too much, if at all.
\begin{lemma}\label{l:inline:helpers}
  Let $l \to r$ be a rewrite rule subject to a redex preserving inlining of function $\funone$ at position $\posone$ in $r$.
  Suppose that the symbol $\funone$ is sufficiently defined by $\atrsone$.
  Consider a normalising reduction 
  \[
    \fmain(\seq[n]{\dataone}) \rss[\atrsone] C[l\sigma] \rew[\atrsone] C[r\sigma] \rsl[\atrsone]{\ell} \termthree \tkom
  \] 
  for $\dataone_i \in \DATA$ ($i = 1,\dots,n$) and some $\ell \in \N$. 
  Then there exists a term $\termone$ such that the following properties hold:
  \begin{enumerateenv}
    \item\label{l:inline:helpers:sd}
      $l\sigma \nR \termone$; and 
    \item\label{l:inline:helpers:step}
      $r\sigma \fR \termone$, 
      where $\atrsI$ collects all rules that are unifiable with the left-hand side $r$ at position $\posone$; and 
    \item\label{l:inline:helpers:rp}
      $C[\termone] \rsl[\atrsone]{{\geqslant} \ell - 1} \termthree$.
  \end{enumerateenv}
\end{lemma}
\longv{
\begin{proof}
  Consider the first property, under the assumptions of the lemma.
  Since $\funone$ is sufficiently defined, the subterm $\subtermAt{r}{\posone}\sigma$ of $r\sigma$
  rooted in $\funone$ is a redex, in particular, $\subtermAt{r}{\posone}\sigma$ matches the 
  left-hand side $u$ of a rule
  $u \to v \in \atrsI$, say $\subtermAt{r}{\posone}\sigma = u\tau$ for some substitution $\tau$.
  Wlog.\ we suppose that the rules in $\atrsone$ are variable disjoint with $r$.
  Hence $\sigma \uplus \tau$ is a well-defined unifier of $\subtermAt{r}{\posone}$ and $u$.
  Let $\mguone$ be a most general unifier of $\subtermAt{r}{\posone}$ and $u$.
  We thus have a substitution $\sigma'_\mguone$ such that for all variables $\varone$ in $\subtermAt{r}{\posone}$,
  $\sigma(\varone) = \mguone(\varone)\sigma'_\mguone$ holds.
  Let $\sigma_\mguone$ be the least extension of $\sigma'_\mguone$ such that $\sigma_\mguone(\varone) = \sigma(\varone)$
  for variables in $l$ which do not occur in $\subtermAt{r}{\posone}$.
  We conclude
  $l\sigma = (l\mguone)\sigma_\mguone \nR (\ctx[\posone]{r\mguone}[v\mguone])\sigma_\mguone$, 
  where the equality follows by definition of $\sigma_\mguone$, and the step by definition of $\narrowings{\posone}{l \to r}$.
  The property follows by taking $\termone = (\ctx[\posone]{r\mguone}[v\mguone])\sigma_\mguone$.

  Now for the second property, 
  recall $l\sigma = (l\mguone)\sigma_\mguone \rew[\atrsone] (r\mguone)\sigma_\mguone = r\sigma$.
  Let $D$ denote the context obtained by replacing the subterm at 
  position $\posone$ in $r\sigma$ by the hole $\hole$, 
  hence 
  \[
    D = \ctx[\posone]{r\sigma}[\hole] 
    = \ctx[\posone]{(r\mguone)\sigma_\mguone}[\hole] 
    = (\ctx[\posone]{r\mguone}[\hole])\sigma_\mguone \tpkt
  \]
  Since $\mguone$ is an mgu of $\subtermAt{r}{\posone}$ and $u$, we thus have
  \[
    r\sigma = D[(\subtermAt{r}{\posone}\mguone)\sigma_\mguone] = D[(u\mguone)\sigma_\mguone]\tpkt
  \]
  Then it is not difficult to conclude that $r\sigma \fR D[(v\mguone)\sigma_\mguone] = \termone$, 
  using that $u \to v \in \atrsI$ is redex preserving wrt.\ the considered inlining and that $\sigma_\mguone$ contain no defined symbols.

  For the final property, consider the sequence $C[r\sigma] \rsl[\atrsone]{\ell} \termthree$, for $\termthree$ in normalform.
  As we observed before, $r\sigma = D[u\tau]$ for the context $D$ defined above, $u \in v \in \atrsI$ and $\tau$ a substitution.
  Using Lemma~\ref{l:narrow:irew}, and employing that redexes are non-overlapping by assumption on $\atrsone$, 
  we can thus obtain an alternate derivation of equal length, where we first completely reduce $u\tau$:
  \[
    C[r\sigma] = C[D[u\tau]] \irsl[\atrsone]{\ell_1} C[D[u\tau_n]] \rew[\atrsone] C[D[v\tau_n]] \irsl[\atrsone]{\ell_2} \termthree \tpkt
  \]
  Here, $\tau_n$ is the normalised substitution obtained by normalising $\tau$, and $\ell = \ell_1 + \ell_2 + 1$. 
  Note that by construction, we have $t = D[v\tau]$. Guided by the above derivation we see
  \[
    C[t] = C[D[v\tau]] \irsl[\atrsone]{k_1} C[D[v\tau_n]] \irsl[\atrsone]{\ell_2} \termthree \tpkt
  \]
  Using that the step $D[u\tau] \fR D[v\tau]$ is not deleting redexes occurring in the substitution $\tau$
  by definition, we have $k_1 \geqslant \ell_1$. 
  In total, the last sequence is thus of length $k_1 + \ell_2 \geqslant \ell_1 + \ell_2$. From the definition of $\ell$, 
  the last property follows. 
\end{proof}
}
\noindent In consequence, we thus obtain a term $\termone$
\[
  \fmain(\seq[n]{\dataone}) \rss[\atrsone] C[l\sigma] \nR C[\termone] \rsl[\atrsone]{{\geqslant} \ell - 1} \termthree \tkom
\] 
holds under the assumptions of the lemma.
Complexity preservation of inlining, modulo a constant factor under the outlined assumption, now 
follows essentially by induction on the maximal length of reductions. 
As a minor technical complication, we have to consider the broader reduction relation 
$\frew[\atrsone]$ instead of $\rew[\atrsone]$. To ensure that the induction 
is well-defined, we use the following specialization of {\cite[Theorem 3.13]{Gramlich:FI:95}}.
\begin{proposition}\label{p:gramlich:95}
  If a term $\termone$ has a normalform wrt. $\rew[\atrsone]$,
  then \emph{all} $\frew[\atrsone]$ reductions of $\termone$ are finite.
\end{proposition}

\begin{theorem}\label{t:inlining}
  Let $l \to r$ be a rewrite rule subject to a redex preserving inlining of function $\funone$ at position $\posone$ in $r$.
  Suppose that the symbol $\funone$ is sufficiently defined by $\atrsone$.
  Let $\atrstwo$ be obtained by replacing rule $l \to r$ by the rules $\narrowings{\posone}{l \to r}$. 
  Then every normalizing derivation with respect to $\atrsone$ starting from
  $\fmain(\seq[n]{\dataone})$ ($\dataone_j \in \DATA$) of
  length $\ell$ is simulated by a derivation with respect to $\atrstwo$
  from $\fmain(\seq[n]{\dataone})$ of length at least $\roundbelow{\frac{\ell}{2}}$.
\end{theorem}
\begin{proof}
  Suppose $\termone$ is a reduct of $\fmain(\seq[n]{\dataone})$ occurring in a normalising reduction, i.e., 
  $\fmain(\seq[n]{\dataone}) \rss[\atrsone] \termone \rss[\atrsone] \termthree$, for $\termthree$ 
  a normal form of $\atrsone$. 
  In proof, we show if $\termone \rss[\atrsone] \termthree$ is a derivation of length $\ell$, 
  then there exists a normalising derivation with respect to $\atrstwo$ whose length is at least $\roundbelow{\frac{\ell}{2}}$.
  The theorem then follows by taking $\termone = \fmain(\seq[n]{\dataone})$.

  We define the \emph{derivation height} $\dheight(\termtwo)$ of a term $\termtwo$ wrt.\ the relation $\frew[\atrsone]$
  as the maximal $m$ such that $\termone \frew[\atrsone]^m \termthree$ holds.
  The proof is by induction on $\dheight(\termone)$, which is well-defined by assumption and Proposition~\ref{p:gramlich:95}.
  It suffices to consider the induction step.
  Suppose $\termone \rew[\atrsone] \termtwo \rsl[\atrsone]{\ell} \termthree$.
  We consider the case where the step $\termone \rew[\atrsone] \termtwo$
  is obtained by applying the rule $l \to r \in \atrsone$, otherwise,
  the claim follows directly from induction hypothesis.
  Then as a result of Lemma~\eref{l:inline:helpers}{sd} and \eref{l:inline:helpers}{rp} we obtain 
  an alternative derivation 
  \[
    \termone \rew[\atrstwo] \termtwo' \rsl[\atrsone]{\ell'} \termthree\tkom
  \]
  for some term $\termtwo'$ and $\ell'$ satisfying $\ell' \geqslant \ell - 1$.
  Note that $\termtwo \frew[\atrsone] \termtwo'$ as a consequence 
  of Lemma~\eref{l:inline:helpers}{step}, and thus
  $\dheight(\termtwo) > \dheight(\termtwo')$ by definition of derivation height.
  Induction hypothesis on $\termtwo'$ thus yields a derivation $\termone \rew[\atrstwo] \termtwo' \rss[\atrstwo] \termthree$
  of length at least $\roundbelow{\frac{\ell'}{2}} + 1 = \roundbelow{\frac{\ell'+2}{2}} \geqslant \roundbelow{\frac{\ell + 1}{2}}$.
\end{proof}

We can then obtain that inlining has the key property we
require on transformations. 

\begin{corollary}[Inlining Transformation]\label{c:narrowing}
  The \emph{inlining transformation}, which replaces a rule $l \to r \in \atrsone$
  by $\narrowings{\posone}{l \to r}$,
  is asymptotically complexity reflecting whenever the function considered for
  inlining is sufficiently defined and the inlining itself is redex preserving.
\end{corollary}

\begin{example}\label{ex:narrowing}
  Consider the ATRS $\atrsrev$ from Section~\ref{s:defunc}.
  Three applications of inlining result in the following ATRS:
\begin{lstlisting}[style=atrs,style=numbers]
C1(f,g) @ z -> f @ (g @ z) %\label{atrsrev':C1}%
C2 @ z -> z                %\label{atrsrev':C2}%
C3(x) @ z -> x::z          %\label{atrsrev':C3}%
comp1(f) @ g -> C1(f,g)
comp @ f -> comp1(f)
match[walk]([]) -> C2
match[walk](x::ys) ->      %\surpressln%
  comp @ (fix[walk] @ ys) @ C3(x) %\showln%
walk @ xs -> match[walk](xs)      %\label{atrsrev':walk}%
fix[walk] @ [] -> C2              %\label{atrsrev':fixnil}%
fix[walk] @ (x::ys) ->            %\label{atrsrev':fix}\surpressln%
  C1(fix[walk] @ ys,C3(x))        %\showln%
rev @ l -> fix[walk] @ l @ []     %\label{atrsrev':rev}%
main(l) -> fix[walk] @ l @ []     %\label{atrsrev':main}%
\end{lstlisting}
The involved inlining rules are all non-erasing,
i.e., all inlinings are \emph{redex preserving}. As a consequence of Corollary~\ref{c:narrowing},
a bound on the runtime complexity of the above system can be relayed,
within a constant multiplier, back to the ATRS $\atrsrev$.
\end{example}
Note that the modified system from Example~\ref{ex:narrowing} gives
further possibilities for inlining. For instance, we could narrow
further down the call to \atrs|fix[walk]| in
rules~\eqref{atrsrev':fix},~\eqref{atrsrev':rev}
and~\eqref{atrsrev':main}, performing case analysis on the variable
\atrs|ys| and \atrs|l|, respectively.  Proceeding this way would
step-by-step unfold the definition of \atrs|fix[walk]|, \emph{ad infinitum}.
We could have also further reduced the rules defining
\atrs|match[walk]| and \atrs|walk|. However, it is not difficult to see
that these rules will never be unfolded in a call to \atrs|main|, they
have been sufficiently inlined and can be removed.  Elimination of such
\emph{unused} rules will be discussed next.
\subsection{Elimination of Dead Code}\label{s:simpl:de}
The notion of \emph{usable rules} is well-established in the rewriting
community. Although its precise definition depends on the context used
(e.g.\ termination~\cite{AG:TCS:00} and
complexity analysis~\cite{HM:IJCAR:08}), the notion commonly refers to
a syntactic method for detecting that
certain rules can never be applied in derivations starting from a
given set of terms. From a programming point of view, such rules
correspond to \emph{dead code}, which can be safely eliminated.

Dead code arises frequently in automatic program transformations, and
its elimination turns out to facilitate our transformation-based
approach to complexity analysis. The following proposition formalises
\emph{dead code elimination} abstractly, for now.  Call a rule $l \to
r \in \atrsone$ \emph{usable} if it can be applied in a derivation
\[
\atrs|main|(\seq[k]{\dataone}) \rew[\atrsone] \cdots \rew[\atrsone] \termone_1 \rew[\{l \to r\}]  \termone_2 \tkom
\]
where $\dataone_i \in \DATA$. The rule $l \to r$ is \emph{dead code}
if it is not usable. The following proposition follows by definition.
\begin{proposition}[Dead Code Elimination]\label{p:usablerules}
  \emph{Dead code elimination}, which maps an ATRS $\atrsone$ to
  a subset of $\atrsone$ by removing dead code only,
  is complexity reflecting and preserving.
\end{proposition}

It is not computable in general which rules are dead code.
One simple way to eliminate dead code is to collect all
the function symbols underlying the definition of \atrs|main|,
and remove the defining rules of symbols not in this collection,
compare e.g.~\cite{HM:IJCAR:08}.
This approach works well for standard TRSs, but is usually
inappropriate for ATRSs where most rules define a single function symbol,
the application symbol.
A conceptually similar, but unification based, approach that works reasonably well
for ATRSs is given in~\cite{GTSK:FROCOS:05}.
However, the accurate identification of dead code, in particular in the
presence of higher-order functions,
requires more than just a simple syntactic analysis.
We show in Section~\ref{s:impl:cfa} a particular form of control flow analysis which
leverages dead code elimination. The following example indicates that such an analysis
is needed.

\begin{example}\label{ex:usablerules}
  We revisit the simplified ATRS from Example~\ref{ex:narrowing}.
  The presence of the composition rule~\eqref{atrsrev':C1},
  itself a usable rule, makes it harder to infer which of the application rules are dead code.
  Indeed, the unification-based method found in~\cite{GTSK:FROCOS:05}
  classifies all rules as usable.
  As we hinted in Section~\ref{s:defunc}, the variables \atrs|f| and \atrs|g| are instantiated
  only by a very limited number of closures in a call of \atrs|main(l)|. In particular, none of the symbols
  \atrs|rev|, \atrs|walk|, \atrs|comp| and \atrs|comp1| are passed to \atrs|C1|.
  With this knowledge, it is not difficult to see that their defining rules, 
  together with the rules defining \atrs|match[walk]|, can be eliminated
  by Proposition~\ref{p:usablerules}.
  Overall, the complexity of the ATRS depicted in
  Example~\ref{ex:narrowing} is thus reflected by the ATRS consisting of the following six rules.
\begin{lstlisting}[style=atrs,style=numbers]
C1(f,g) @ x -> f @ (g @ x)     %\label{atrsrev'':C1}%
C2 @ z -> z
C3(x) @ z -> x::z
fix[walk] @ [] -> C2
fix[walk] @ (x::ys) -> %\surpressln%
  C1(fix[walk] @ ys$\!$,$\,$C3(x)) %\showln%
main(l) -> fix[walk] @ l @ []
\end{lstlisting}
\end{example}
\subsection{Instantiation}
Inlining and dead code elimination can indeed help in
simplifying defunctionalised programs. There is however a feature of
ATRS they cannot eliminate in general, namely rules whose right-hand-sides have
\emph{head variables}, i.e., variables that occur to the left of an 
application symbol and thus denote a function. 
The presence of such rules prevents FOPs to succeed in all but trivial
cases.
The ATRS from Example~\ref{ex:usablerules}, for instance, still
contains one such rule, namely rule~\eqref{atrsrev'':C1},
with head variables \atrs|f| and \atrs|g|.
The main reason FOPs perform poorly on ATRS
containing such rules is that they lack any form of control
flow analysis, and they are thus unable to realise that function
symbols simulating higher-order combinators are passed arguments of a
very specific shape, and are thus often harmless. This is the case, as
an example, for the function symbol \atrs|C1|.

The way out consists in specialising the ATRS rules. This has the
potential of highlighting the \emph{absence} of certain dangerous
patterns, but of course must be done with great care, without hiding
complexity under the carpet of non-exhaustive instantiation. All this
can be formalised as follows.

\newcommand{\nt}[1]{\mathsf{#1}}
\newcommand{\ntany}{\star}
\newcommand{\Rl}[1]{R_#1}
\newcommand{\V}[2]{#1_#2}

Call a rule $l' \to r'$ an \emph{instance} of a rule $l \to r$, if
there is a substitution $\sigma$ with $l' = l\sigma$ and $r' = r\sigma$.
We say that an ATRS $\atrstwo$ is an
\emph{instantiation} of $\atrsone$ iff all rules in $\atrstwo$ are
instances of rules from $\atrsone$. This instantiation is
\emph{sufficiently exhaustive} if for every derivation
\[
\atrs|main|(\seq[k]{\dataone}) \rew[\atrsone] \termone_1 \rew[\atrsone] \termone_2 \rew[\atrsone] \cdots \tkom
\]
where $\dataone_i \in \DATA$, there exists a corresponding derivation
\[
\atrs|main|(\seq[k]{\dataone}) \rew[\atrstwo] \termone_1 \rew[\atrstwo] \termone_2 \rew[\atrstwo] \cdots  \tpkt
\]
The following proposition is immediate from the definition.
\begin{theorem}[Instantiation Transformation]\label{t:instantiate}
  Every \emph{instantiation transformation}, mapping any ATRS into a
  sufficiently exhaustive instantiation of it, is complexity reflecting
  and preserving.
\end{theorem}

\begin{example}[Continued from Example~\ref{ex:usablerules}]\label{ex:instantiate}
We instantiate the rule \atrs|C1(f,g) @ x -> f @ (g @ x)|
by the two rules
\begin{lstlisting}[style=atrs,style=numbers]
C1(C2,C3(x)) @ z -> C3(x) @ (C2 @ z)
C1(C1(f,g),C3(x)) @ z -> %\surpressln%
  C1(f,g) @ (C2 @ z) %\showln%
\end{lstlisting}
leaving all other rules from the TRS depicted in Example~\ref{ex:usablerules} intact.
As we reasoned already before, the instantiation is sufficiently exhaustive:
in a reduction of \atrs|main(l)| for a list \atrs|l|,
arguments to \atrs|C1| are always of the form as indicated in the two rules.
Note that the right-hand side of both rules can be reduced by inlining the calls
in the right argument.
Overall, we conclude that the runtime complexity of our running example is reflected
in the ATRS consisting of the following six rules:
\begin{lstlisting}[style=atrs,style=numbers]
C2 @ z -> z
C1(C2,C3(x)) @ z -> x::z
C1(C1(f,g),C3(x)) @ z -> %\surpressln%
  C1(f,g) @ (x::z) %\showln%
fix[walk] @ [] -> C2 %\label{atrs'':fix:bc}%
fix[walk] @ (x::ys) -> %\surpressln%
  C1(fix[walk] @ ys,C3(x)) %\showln\label{atrs'':fix:rec}%
main(l) -> fix[walk] @ l @ [] %\label{atrs'':main}%
\end{lstlisting}
\end{example}
\subsection{Uncurrying}\label{s:simpl:uncurrying}
The ATRS from Example~\ref{ex:instantiate} is now sufficiently
instantiated: for all occurrences of the \atrs|@| symbol, we know
\emph{which function} we are applying, even if we do not necessarily
know \emph{to what} we are applying it. The ATRS is not yet ready
to be processed by FOPs, simply because the application function symbol
is anyway still there, and cannot be dealt with.

At this stage, however, the ATRS can indeed be brought to a form
suitable for analysis by FOPs through 
\emph{uncurrying}, see e.g.\ the account of~\citet{HMZ:JAR:13}.
Uncurrying an ATRS $\atrsone$ involves the definition of a fresh function symbol
$\funone^n$ for each $n$-ary application
\[
\funone(\seq[m]{\termone}) \api \termone_{m+1} \api \cdots \api \termone_{m+n}
\tkom
\]
encountered in $\atrsone$. This way, applications can be completely
eliminated.  Although in~\cite{HMZ:JAR:13} only ATRSs defining
function symbols of null arity are considered, the extension to our
setting poses no problem. We quickly recap the central definitions.

Define the \emph{applicative arity} $\aarity{\atrsone}{\funone}$ of a symbol
$\funone$ in $\atrsone$
as the maximal $n \in \N$ such that a term 
\[
  \funone(\seq[m]{\termone}) \api \termone_{m+1} \api \cdots \api \termone_{m+n} \tkom
\] 
occurs in $\atrsone$.
%
\begin{definition}
  The \emph{uncurrying} $\uncurry{\termone}$ of a
  term $\termone = \funone(\seq[m]{\termone}) \api \termone_{m+1} \api \cdots \api \termone_{m+n}$,
  with $n \leqslant \aarity{\atrsone}{\funone}$ is defined as
  \[
  \uncurry{\termone} \defsym \funone^n(\uncurry{\termone_1},\dots,\uncurry{\termone_m},\uncurry{\termone_{m+1}},\dots,\uncurry{\termone_{m+n}})
  \tkom
  \]
  where $\funone^0 = \funone$ and $\funone^n$ ($1 \leq n \leq \aarity{\atrsone}{\funone}$) are fresh function symbols.
  Uncurrying is homomorphically extended to ATRSs.
\end{definition}

Note that
$\uncurry{\atrsone}$ is well-defined whenever $\atrsone$ is \emph{head variable free}, i.e.,
does not contain a term of the form \atrs|x @  $t$| for variable \atrs|x|.
We intend to use the TRS $\uncurry{\atrsone}$ to simulate reductions
of the ATRS $\atrsone$.
In the presence of rules of functional type however,
such a simulation fails.
To overcome the issue, we \emph{$\eta$-saturate} $\atrsone$.
\begin{definition}
  We call a TRS $\atrsone$ \emph{$\eta$-saturated} if whenever 
  \[
    \funone(\seq[m]{l}) \api l_{m+1} \api \cdots \api l_{m+n} \to r \in \atrsone \text{ with $n < \aarity{\atrsone}{\funone}$,}
  \]
  then it contains also a rule 
  \[
    \funone(\seq[m]{l}) \api l_{m+1} \api \cdots \api l_{m+n} \api {\atrs|z|} \to r \api {\atrs|z|} \tkom
  \]
  where \atrs|z| is a fresh variable.
  The \emph{$\eta$-saturation} $\etasaturate{\atrsone}$ of $\atrsone$ is defined as the least extension of $\atrsone$
  that is $\eta$-saturated.
\end{definition}
\begin{remark}
  The $\eta$-saturation $\etasaturate{\atrsone}$ of an ATRS $\atrsone$ is 
  not necessarily finite. 
  A simple example is the one-rule ATRS $\funone \to \funone \api \fun{a}$ where 
  both $\funone$ and $\fun{a}$ are function symbols. 
  Provided that the ATRS $\atrsone$ is endowed with simple types, and indeed 
  the simple typing of our initial program is preserved throughout our complete transformation 
  pipeline, the $\eta$-saturation of $\atrsone$ becomes finite. 
\end{remark}

\begin{example}[Continued from Example~\ref{ex:instantiate}]\label{ex:saturated}
  The ATRS from Example~\ref{ex:instantiate} is not $\eta$-saturated: 
  \atrs|fix[walk]| is applied to two arguments in rule~\eqref{atrs'':main}, 
  but its defining rules, rule~\eqref{atrs'':fix:bc} and~\eqref{atrs'':fix:rec}, 
  take a single argument only. 
  The $\eta$-saturation
  thus contains in addition the following two rules:
\begin{lstlisting}[style=atrs, style=numbers]
fix[walk] @ [] @ z -> C2 @ z
fix[walk] @ (x::ys) @ z -> %\surpressln%
  C1(fix[walk] @ ys,C3(x)) @ z . %\showln%
\end{lstlisting}
One can then check that the resulting system is $\eta$-saturated.
\end{example}

\begin{lemma}\label{l:uncurry:aux}
  Let $\etasaturate{\atrsone}$ be the $\eta$-saturation of $\atrsone$.
  \begin{enumerateenv}
    \item\label{l:uncurry:aux:eta} The rewrite relation $\rew[\etasaturate{\atrsone}]$ coincides with $\rew[\atrsone]$.
    \item\label{l:uncurry:aux:step} 
      Suppose $\etasaturate{\atrsone}$ is head variable free.
      If $\termtwo \rew[\etasaturate{\atrsone}] \termone$ then $\uncurry{\termtwo} \rew[\uncurry{\etasaturate{\atrsone}}] \uncurry{\termone}$.
  \end{enumerateenv}
\end{lemma}
\begin{proof}
  For Property~\ref{l:uncurry:aux:eta}, the inclusion
  ${\rew[\atrsone]} \subseteq {\rew[\etasaturate{\atrsone}]}$ follows
  trivially from the inclusion $\atrsone \subseteq
  \etasaturate{\atrsone}$. The inverse inclusion ${\rew[\atrsone]}
  \supseteq {\rew[\etasaturate{\atrsone}]}$ can be proven by a
  standard induction on the derivation of $l \to r \in
  \etasaturate{\atrsone}$.  

  Property~\ref{l:uncurry:aux:step} can be
  proven by induction on $\termone$.
  The proof follows the pattern of the proof of~\citet{ST:FROCOS:11}. 
  Notice that in~\cite[Theorem~10]{ST:FROCOS:11},
  the rewrite system $\uncurry{\etasaturate{\atrsone}}$ is enriched with uncurrying rules of the form
  $\funone^i(\seq[n]{\varone}) \api \vartwo \to \funone^{i+1}(\seq[n]{\varone},\vartwo)$. Such an extension is
  not necessary in the absence of head variables. In our setting, the application symbol is completely eliminated
  by uncurrying, and thus the above rules are dead code.
\end{proof}

As a consequence, we immediately obtain the following theorem.
\begin{theorem}[Uncurrying Transformation]\label{t:uncurrying}
  Suppose that $\etasaturate{\atrsone}$ is head variable free.
  The \emph{uncurrying} transformation, which maps an ATRS $\atrsone$
  to the system $\uncurry{\etasaturate{\atrsone}}$, is complexity reflecting.
\end{theorem}

\begin{example}[Continued from Example~\ref{ex:saturated}]\label{ex:uncurry}
  Uncurrying the $\eta$-saturated ATRS, consisting of the six rules from Example~\ref{ex:instantiate}
  and the two rules from Example~\ref{ex:saturated}, results in the 
  following set of rules:
\begin{lstlisting}[style=atrs,style=numbers]
C1u(C2,C3(x),z) -> x::z %\label{uc:C1:1}%
C1u(C1(f,g),C3(x),z) -> C1u(f,g,x::z) %\label{uc:C2:2}%
C2u(z) -> z
fix[walk]u([]) -> C2   %\label{uc:walk:1}%
fix[walk]u(x:ys) -> C1(fix[walk]u(ys),C3(x))%\label{uc:walk:2}%
fix[walk]uu([],z) -> C2u(z)
fix[walk]uu(x::ys,z) -> %\surpressln%
  C1[1](fix[walk]u(ys),C3(x),z) %\showln%
main(l) -> fix[walk]uu(l,[])
\end{lstlisting}
Inlining the calls to \atrs|fix[walk]uu| and \atrs|C2u(z)|, followed by dead code elimination,
results finally in the TRS $\trsrev$ from Section~\ref{s:defunc}.
\end{example}


%% file: implementation.tex
\newcommand\RL[1]{\ensuremath{R_{#1}\,}}
\newcommand\RR[1]{\RL{\ref{atrsrev':#1}}}
\newcommand\VR[1]{\ensuremath{{}_{\ref{atrsrev':#1}}\,}}
\section{Automation}\label{s:auto}
In the last section we have laid the formal foundation of our
program transformation methodology, and ultimately of our
tool \hoca. Up to now, however, program transformations (except for
uncurrying) are too abstract to be turned into actual algorithms. In
dead code elimination, for instance, the underlying computation problem
(namely the one of precisely isolating usable rules) is 
undecidable. In inlining, one has a decidable transformation, which
however results in a blowup of program sizes, if blindly applied.

This section is devoted to describing some \emph{concrete} design choices
we made when automating our program transformations. Another, related, issue
we will talk about is the effective \emph{combination} of these
techniques, the \emph{transformation pipeline}.

\subsection{Automating Inlining}\label{s:impl:inline}
The main complication that arises while automating our inlining
transformation is to decide \emph{where} the transformation should be
applied. Here, there are two major points to consider: 
first, we want to ensure that the overall transformation is not only
complexity \emph{reflecting}, but also complexity \emph{preserving},
thus not defeating its purpose.  To address this issue, we
employ inlining conservatively, ensuring that inlining does not
duplicate function calls.  Secondly, as we already hinted after
Example~\ref{ex:narrowing}, exhaustive inlining is usually not desirable
and may even lead to non-termination in the transformation pipeline
described below. Instead, we want to ensure that inlining
\emph{simplifies} the problem with respect to some sensible
\emph{metric}, and plays well in conjunction with the other
transformation techniques.

  Instead of working with a closed
  inlining strategy, our implementation \strat|inline($P$)| is parameterised by
  a predicate $P$ which, intuitively, tells when inlining 
  a call at position $\posone$ in a rule $l \to r$ is \emph{sensible} at the current
  stage in our transformation pipeline. 
  The algorithm \strat|inline($P$)| replaces every rule $l \to r$ by
  $\narrowings{\posone}{l \to r}$ for some position $\posone$ such
  that $P(\posone,l \to r)$ holds. The following four predicates turned out to be useful
  in our transformation pipeline. 
  The first two are designed by taking into account the specific shape 
  of ATRSs obtained by defunctionalisation, 
  the last two are generic. 
\begin{varitemize}
\item \strat|match|:
  This predicate holds if the right-hand side $r$ 
  is labeled by a symbol of the form $\ccase{cs}$ at position $\posone$. 
  That is, the predicate enables inlining of calls resulting from the translation of a match-expression, 
  thereby eliminating one indirection due to the encoding of pattern matching during defunctionalization.
\item \strat|lambda-rewrite|:
  This predicate holds if the subterm $\subtermAt{r}{\posone}$ is of the 
  form $\clam{\varone}{\expone}(\vec{\termone}) \api \termtwo$.
  Note that by definition it is enforced that inlining corresponds to a plain rewrite, 
  head variables are not instantiated. 
  For instance, \strat|inline(lambda-rewrite)| is inapplicable
  on the rule ${\atrs|C2(f,g) @ z -> f @ (g @ z)|}$.
  This way, we avoid that variables \atrs|f| and \atrs|g| are improperly 
  instantiated.
\item \strat|constructor|: The predicate holds if the right-hand sides
  of all rules used to inline $\subtermAt{r}{\posone}$ are constructor terms, i.e.,
  do not give rise to further function calls. Overall, the number of function calls
  therefore decreases. As a side effect, more patterns become obvious in rules, which facilitates
  further inlining. 
\item \strat|decreasing|: The predicate holds if any of the following two conditions is satisfied:
  (i) \emph{proper inlining}: 
    the subterm $\subtermAt{r}{p}$ constitutes the only call-site to the inlined function $\funone$. 
    This way, all rules defining $\funone$ in $\atrsone$ will turn to dead code after inlining. 
  (ii) \emph{size decreasing}: each right-hand side in $\narrowings{\posone}{l \to r}$ is strictly smaller in size than the right-hand side $r$. 

  Here, the aim is to facilitate FOPs on the generated output. 
  In the first case, the number of rules decreases, which usually implies that in the 
  analysis, a FOP generates less constraints which have to be solved. 
  In the second case, the number of constraints might increase, 
  but the individual constraints are usually easier to solve, due to the decrease 
  in sizes of right hand sides. 
\end{varitemize}
We emphasise that all inlinings performed on our running example $\atrsrev$ are 
captured by the instances of inlining just defined. 
 
\subsection{Automating Instantiation and Dead Code Elimination via Control Flow Analysis}\label{s:impl:cfa}

One way to effectively eliminate dead code and apply instantiation,
as in Examples~\ref{ex:usablerules} and~\ref{ex:instantiate}, consists in inferring the shape of closures
passed during reductions. This way, we can on the one hand specialise
rewrite rules being sure that the obtained instantiation is sufficiently exhaustive,
and on the other hand discover that certain rules are simply useless,
and can thus be eliminated.

To this end, we rely on an approximation of the collecting semantics.
In static analysis, the collecting semantics of a program maps a given
program point to the collection of states attainable when control
reaches that point during execution.  In the context of rewrite
systems, it is natural to define the rewrite rules as program points,
and substitutions as states.  Throughout the following, we fix an ATRS
$\atrsone = \{l_i \to r_i\}_{i\in\{1,\dots,n\}}$.  We define the
\emph{collecting semantics of} $\atrsone$ as a tuple
$(Z_1,\dots,Z_n)$, where
\begin{multline*}
  Z_i \defsym \{ (\sigma,\termone)
        \mid \exists \vec{\dataone} \in \DATA.\ \shortv{\\}
        {\atrs|main|}(\vec{\dataone}) \rss[\atrsone] \ctxone[l_i\sigma] \rew[\atrsone] \ctxone[r_i\sigma] 
        \text{ and $r_i\sigma \rss[\atrsone] \termone$} \} \tpkt
\end{multline*}
Here the substitutions $\sigma$ are restricted to the set $\Var{l_i}$
of variables occurring in the left-hand side in $l_i$.

The collecting semantics of $\atrsone$ includes all the necessary
information for implementing both dead code elimination and instantiation:
\begin{lemma}\label{l:colsem:aux}
  The following properties hold:
  \begin{varenumerate}
  \item\label{l:colsem:aux:deadcode}
    The rule $l_i \to r_i \in \atrsone$ constitutes dead code if and only if~$Z_i = \varnothing$.
  \item\label{l:colsem:aux:instantiate} 
    Suppose the ATRS $\atrstwo$ is obtained by instantiating rules
    $l_i \to r_i$ with substitutions
    $\sigma_{i,1},\dots,\sigma_{i,k_i}$.  Then the instantiation is
    sufficiently exhaustive if for every substitution $\sigma$ with
    $(\sigma,\termone) \in Z_i$, there exists a substitution
    $\sigma_{i,j}$ $(j\in\{1,\dots,i_k\})$ which is at least as
    general as $\sigma$.
  \end{varenumerate}
\end{lemma}
\begin{proof}
  The first property follows by definition.
  For the second property, consider a derivation
  \[
  \atrs|main|(\seq[k]{\dataone}) \rss[\atrsone] C[l_i\sigma] \rew[\atrsone] C[r_i\sigma] \tkom
  \]
  and thus $(\sigma,r_i\sigma) \in Z_i$. 
  By assumption, there exists a substitution $\sigma_{i,j}$ ($i \in \{1,\dots,i_k\}$) is at least as general as $\sigma$.
  Hence the ATRS $\atrstwo$ can simulate the step from $C[l_i\sigma] \rew[\atrsone] C[r_i\sigma]$, 
  using the rule $l_i\sigma_{i,j} \to r_i\sigma_{i,j} \in \atrstwo$. 
  From this, the property follows by inductive reasoning.
\end{proof}
As a consequence, the collecting semantics of $\atrsone$ is itself
not computable. Various techniques to over-approximate the
collecting semantics have been proposed, e.g.\ by
\citet{FGT:2004}, \citet{Jones:TCS:07} and \citet{KO:RTA:11}.
In all the works above, the approximation consists in describing
the tuple $(Z_1,\dots,Z_n)$ by a \emph{finite} object.

\def\aut{\lstinline[style=atrs,mathescape]}
\begin{figure*}[!t]
\fbox{
\begin{minipage}[t]{0.975\linewidth}
\vspace{-2mm}
\centering
  \begin{minipage}[t]{0.40\linewidth}
    \begin{align*}
      {\aut!*!} \to{} & {\aut![]!} \mid {\aut!*::*!} \\
      S \to{} & {\aut|main(*)|} \mid \RR{main}\\
      \RR{main} \to{} & {\aut!fix[walk] @ l$\VR{main}$ @ []!} \mid {\aut!$\RR{fixnil}$ @ []!} \mid {\aut!$\RR{fix}$ @ []!} \mid \RR{C1} \mid \RR{C2} \\
      \mid{} & {\aut!C1($\RR{fixnil}$,C3(x$\VR{fix}$))!} \mid {\aut!C1($\RR{fix}$,C3(x$\VR{fix}$))!}\\
      \RR{fix} \to{} & {\aut!C1(fix[walk] @ ys$\VR{fix}$,C3(x$\VR{fix}$))!} \\
      \RR{fixnil} \to{} & {\aut!C2!}
    \end{align*}
  \end{minipage}
  \begin{minipage}[t]{0.35\linewidth}
    \begin{align*}
      \RR{C1} \to{} & {\aut!f$\VR{C1}$ @ (g$\VR{C1}$ @ z$\VR{C1}$)!} \mid {\aut!f$\VR{C1}$ @ $\,\RR{C3}$!} \mid \RR{C1} \mid \RR{C2}\\
      \RR{C3} \to{} & {\aut!x$\VR{C3}$::z$\VR{C3}$!} \\
      \RR{C2}  \to{} & {\aut!z$\VR{C2}$!}\\
      {\aut!l$\VR{main}$!} \to{} & {\aut!*!} \\
      {\aut!x$\VR{fix}$!} \to{} & {\aut!*!} \\ 
      {\aut!ys$\VR{fix}$!} \to{} & {\aut!*!}
    \end{align*}
  \end{minipage}
  \begin{minipage}[t]{0.15\linewidth}
    \begin{align*}
      {\aut!z$\VR{C3}$!} \to{} & {\aut!z$\VR{C1}$!} \\
      {\aut!f$\VR{C1}$!} \to{} & \RR{fixnil} \mid \RR{fix} \\
      {\aut!g$\VR{C1}$!} \to{} & {\aut!C3(x$\VR{fix}$)!} \\
      {\aut!z$\VR{C1}$!} \to{} & \RR{C3} \mid {\aut![]!} \\
      {\aut!z$\VR{C2}$!} \to{} & \RR{C3} \mid {\aut![]!} \\
      {\aut!x$\VR{C3}$!} \to{} & {\aut!x$\VR{fix}$!}
    \end{align*}
  \end{minipage}
\end{minipage}
}
\shortv{\nocaptionrule}\caption{Over-approximation of the collecting semantics of the ATRS from Example~\ref{ex:narrowing}.}\label{fig:cfa}
\end{figure*}

In \hoca\ we have implemented a variation of the technique of
\citet{Jones:TCS:07}, tailored to call-by-value semantics (already
hinted at in~\cite{Jones:TCS:07}).
Conceptually, the form of control flow analysis we perform is close to 
a $0$-CFA~\cite{NNH:2005},
merging information derived from different call sites. 
Whilst being efficient to compute, the precision of this relatively simple approximation
turned out to be reasonable for our purpose. 

The underlying idea is to construct a \emph{(regular) tree grammar} which
over-approximates the collecting semantics. 
Here, a \emph{tree grammar} $\tgone$ can be seen as a ground
ATRS whose left-hand sides are all function symbols with arity zero. 
The \emph{non-terminals} of
$\tgone$ are precisely the left-hand sides.  
For the remaining, we assume that variables occurring $\atrsone$ 
are indexed by indices of rules, i.e., every variable occuring 
in the \nth{$i$} rule $l_i \to r_i \in \atrsone$ has index $i$. 
Hence the set of variables of rewrite rules in $\atrsone$ 
are pairwise disjoint. 
Besides a designated non-terminal $S$, the \emph{start-symbol}, 
the constructed tree grammar $\tgone$ admits two kinds of non-terminals: 
non-terminals $\RL{i}$ for each rule $l_i \to r_i$ and non-terminals $\atrs|z|_i$
for variables $\atrs|z|_i$ occurring in $\atrsone$. 
Note that the latter 
the variable $\atrs|z|_i$ is considered as a constant in $\tgone$.
We say that $\tgone$ is \emph{safe} for $\atrsone$ if the following
two conditions are satisfied for all $(\sigma,\termone) \in Z_i$:
(i)~$\atrs|z|_i \rss[\tgone] \sigma(\atrs|z|_i)$ for each ${\atrs|z|_i} \in \Var{l_i}$; and 
(ii)~$\RL{i} \rss[\tgone] \termone$.
This way, $\tgone$ constitutes a finite over-approximation of the collecting
semantics of $\atrsone$.

\begin{example}
  Figure~\ref{fig:cfa} shows the tree grammar $\tgone$ constructed by 
  the method described below, which is safe for the ATRS $\atrsone$
  from Example~\ref{ex:narrowing}. The notation $N \to \termone_1 \mid \cdots \mid \termone_n$
  is short-hand for the $n$ rules $N \to \termone_i$. 
\end{example}

The construction of Jones consists of an initial automaton $\tgone_0$, which describes considered start terms, 
and which is then systematically closed under rewriting by way of an extension operator $\tgextension{\cdot}$. 
Suitable to our concerns, we define $\tgone_0$ as the tree grammar consisting of the following rules:
\begin{align*}
  S & \to {\atrs|main|}({\atrs|*|},\dots,{\atrs|*|}) && \text{and}\\
  {\atrs|*|} & \to \conone_j({\atrs|*|},\dots,{\atrs|*|}) && \text{for each constructor $\conone_j$ of $\atrsone$.}
\end{align*}
Then clearly $S \rss[\tgone]{\atrs|main|}(\seq[n]{\dataone})$ for all inputs $\dataone_i \in \DATA$. 
We let $\tgone$ be the least set of rules satisfying 
$\tgone \supseteq \tgone_0 \cup \tgextension{\tgone}$
with 
\[
\tgextension{\tgone} \defsym \bigcup_{N \to C[u] \in \tgone} \tgext{N}{C}{u}\tpkt
\]
Here, $\tgext{N}{C}{u}$ is defined as the following set of rules:
\[
\left\{
  \begin{array}{l | l}
    N \to C[\RL{i}],
    & l_i \to r_i \in \atrsone, \\
    \RL{i} \to r_i, \text{ and}
    & u \rss[\tgone] l_i\sigma \text{ is minimal}\\
    {\atrs|z|}_i \to \sigma({\atrs|z|}) \text{ for all ${\atrs|z|} \in \Var{l_i}$} 
    & \text{and $\sigma$ normalised.}
  \end{array}
\right\}
\]
In contrast to~\cite{Jones:TCS:07}, we require
that the substitution $\sigma$ is normalised, 
thereby modelling call-by-value semantics.
The tree grammar $\tgone$ is computable using a simple 
fix-point construction. Minimality of $\funone(\seq[k]{\termone}) \rss[\tgone] l_i\sigma$ means that 
there is no shorter sequence $\funone(\seq[k]{\termone}) \rss[\tgone] l_i\tau$ 
with $l_i\tau \rss[\tgone] l_i\sigma$, and ensures that $\tgone$ is finite~\cite{Jones:TCS:07},
thus the construction is always terminating.

We illustrate the construction on the ATRS from Example~\ref{ex:narrowing}.
\begin{example}
  Revise the ATRS from Example~\ref{ex:narrowing}.  To construct the
  safe tree grammar as explained above, we start from the initial
  grammar $\tgone_0$ given by the rule
  \begin{align*}
  S \to{} & {\aut|main(*)|} & {\aut!*!} \to{} & {\aut![]!} \mid {\aut!*::*!} \tkom
  \end{align*}
  and then successively fix violations of the above closure condition.
  The only violation in the initial grammar is caused by the first
  production.  Here, the right-hand side {\aut|main(*)|} matches the (renamed)
  rule~\ref{atrsrev':main}: \atrs|main(l) -> fix[walk] @ l$\VR{main}$ @ []|,
  using the substitution $\{ {\aut|l|} \mapsto {\aut|*|} \}$. We fix
  the violation by adding productions
  \begin{align*}
    S \to{} & \RR{main} & 
    \RR{main} \to{} & {\aut!fix[walk] @ l$\VR{main}$ @ []!} & 
    {\aut!l$\VR{main}$!} \to{} & {\aut!*!} \tpkt
  \end{align*}
  The tree grammar $\tgone$ constructed so far tells us that
  $l\VR{main}$ is a list. In particular, we have the following two
  minimal sequences which makes the left subterm of the
  $\RR{main}$-production an instances of the left-hand sides of
  defining rules of \atrs|fix[walk]| (rules~\eqref{atrsrev':fixnil}
  and~\eqref{atrsrev':fix}):
  \begin{align*}
    {\aut!fix[walk] @ l$\VR{main}$!} & \rst[\tgone] {\aut!fix[walk] @ []!} \tkom \\
    {\aut!fix[walk] @ l$\VR{main}$!} & \rst[\tgone] {\aut!fix[walk] @ *::*!} \tkom
  \end{align*}
  To resolve the closure violation, the tree grammar is extended by productions
  \begin{align*}
    \RR{main} \to{} & {\aut!$\RR{fixnil}$  @ []!} & 
    \RR{fixnil} \to{} & {\aut!C2!}
  \end{align*}
  because of rule~\eqref{atrsrev':fixnil}, and by 
  \begin{align*}
    \RR{main} \to{} & {\aut!$\RR{fix}$ @ []!}
    & {\aut!x$\VR{fix}$!} \to{} & {\aut!*!} \\
    \RR{fix} \to{} & {\aut!C1(fix[walk] @ ys$\VR{fix}$,C3(x$\VR{fix}$))!}
    & {\aut!ys$\VR{fix}$!} \to{} & {\aut!*!} \tpkt
  \end{align*}
  due to rule~\eqref{atrsrev':fix}. We can now identify a new
  violation in the production of $\RR{fix}$. Fixing all violations
  this way will finally result in the tree grammar depicted in
  Figure~\ref{fig:cfa}.
\end{example}


The following lemma confirms that $\tgone$ is closed under rewriting with respect to the call-by-value semantics.
The lemma constitutes a variation of Lemma~5.3 from~\cite{Jones:TCS:07}.
\begin{lemma}\label{l:tgaux}
  If $S \rss[\tgone] \termone$ and $\termone \rss[\atrsone] C[l_i\sigma] \rew[\atrsone] C[r_i\sigma]$ then 
  $S \rss[\tgone] C[\RL{i}]$, $\RL{i} \rew[\tgone] r_i$ and 
  $\atrs|z|_i \rss[\tgone] \sigma(\atrs|z|_i)$ for all variables $\atrs|z|_i \in \Var{l_i}$. 
\end{lemma}
\begin{theorem}\label{t:tg:safe}
  The tree grammar $\tgone$ is safe for $\atrsone$. 
\end{theorem}
\begin{proof}
  Fix $(\sigma,\termone) \in Z_i$, and let $\atrs|z| \in \Var{l_i}$. 
  Thus ${\atrs|main|}(\vec{\dataone}) \rss[\atrsone] \ctxone[l_i\sigma] \rew[\atrsone] \ctxone[r_i\sigma]$ and $r_i\sigma \rss[\atrsone] \termone$ 
  for some inputs $\dataone \in \DATA$. 
  As we have $S \rss[\tgone] {\atrs|main|}(\vec{\dataone})$ since $\tgone_0 \subseteq \tgone$, 
  Lemma~\ref{l:tgaux} yields $R_i \rew[\tgone] r_i$ and
  ${\atrs|z|}_i \rss[\tgone] \sigma({\atrs|z|})$, i.e., the second safeness conditions is satisfied. 
  Clearly, $\RL{i} \rew[\tgone] r_i \rss[\tgone] r_i\sigma$. 
  A standard induction on the length of $r_i\sigma \rss[\atrsone] \termone$ 
  then yields $\RL{i}\rss[\atrsone] \termone$, using again Lemma~\ref{l:tgaux}.
\end{proof}

We arrive now at our concrete implementation ${\strat|cfa|}(\atrsone)$
that employs the above outlined call flow analysis to deal with both
dead code elimination and instantiation on the given ATRS $\atrsone$.
The construction of the tree grammar $\tgone$ follows itself closely
the algorithm outlined by \citet{Jones:TCS:07}.  Recall that the
\nth{$i$} rule $l_i \to r_i \in \atrsone$ constitutes dead code if the
\nth{$i$} component $Z_i$ of the collecting semantics of $\atrsone$ is
empty, by Lemma~\eref{l:colsem:aux}{deadcode}.  Based on the
constructed tree grammar, the implementation identifies rule $l_i \to
r_i$ as dead code when $\tgone$ does not define a production $\RL{i}
\to \termone$ and thus $Z_i = \varnothing$.  All such rules are
eliminated, in accordance to Proposition~\ref{p:usablerules}.  On the
remaining rules, our implementation performs instantiation as
follows.  We suppose \emph{$\epsilon$-productions} $N \to M$, for
non-terminals $M$, have been eliminated by way of a standard
construction, preserving the set of terms from non-terminals in
$\tgone$.  Thus productions in $\tgone$ have the form $N \to
\funone(\seq[k]{\termone})$.  Fix a rule $l_i \to r_i \in \atrsone$.
The primary goal of this stage is to get rid of head variables, with
respect to the $\eta$-saturated ATRS $\etasaturate{\atrsone}$, thereby
enabling uncurrying so that the ATRS $\atrsone$ can be brought into
functional form.  For all such head variables \atrs|z|, then, we construct a
set of binders
\[
\{ {\atrs|z|_i} \mapsto \mathsf{fresh}(\funone(\seq[k]{\termone}))
\mid {\atrs|z|_i} \to \funone(\seq[k]{\termone}) \in \tgone \}
\tkom
\]
where the function $\mathsf{fresh}$ replaces non-terminals by fresh variables, 
discarding binders where the right-hand contains defined symbols.
For variables \atrs|z| which do not occur in head positions, we construct such a binder only 
if the production ${\atrs|z|_{i}} \to \funone(\seq[k]{\termone})$ is unique.
With respect to the tree grammar of Figure~\ref{fig:cfa}, head variables \atrs|f|, \atrs|g|
of the rule~\ref{atrsrev':C1}
the implementation generates binders
\begin{align*}
  \{ {\atrs|f|_1} \mapsto {\atrs!C2!}, {\atrs|f|_1} \mapsto {\atrs!C1(f',C3(x'))!} \} \text{ and }
  \{ {\atrs|g|_1} \mapsto {\atrs!C3(x'')!} \}
  \tpkt                                                               
\end{align*}
The product-combination of all such binders gives then a set of substitution
$\{\sigma_{i,1},\dots,\sigma_{i,i_k}\}$ that leads to sufficiently many instantiations
$l_i\sigma_{i,j} \to r_i\sigma_{i,j}$ of rule $l_i \to r_i$, by Lemma~\eref{l:colsem:aux}{instantiate}. 
Our implementation replaces every rule $l_i \to r_i \in \atrsone$ by instantiations constructed this way. 

The definition of binder was chosen to keep the number of computed substitutions minimal, and hence the generated head variable free ATRS small.
Putting things together, we see that the instantiation is sufficiently exhaustive, and thus 
the overall transformation is complexity reflecting and preserving by Theorem~\ref{t:instantiate}.
By \strat|cfaDCE| we denote the variation of \strat|cfa| that performs dead code elimination, but no instantiations. 

\subsection{Combining Transformations}\label{s:strategy}
\begin{figure}
  \centering
\begin{lstlisting}[style=strategy,frame=single]
simplify = simpATRS; toTRS; simpTRS where
  simpATRS = %\strut%
    exhaustive inline(lambda-rewrite);
    exhaustive inline(match);
    exhaustive inline(constructor);
    usableRules
  toTRS = cfa; uncurry; usableRules %\strut%
  simpTRS = 
    exhaustive ((inline(decreasing); %\strut%
                 usableRules) <> cfaDCE)%\strut%
\end{lstlisting}
\shortv{\nocaptionrule}\caption{Transformation Strategy in \protect\hoca.}\label{fig:strategy}
\end{figure}
We have now seen all the building blocks underlying our tool \hoca.\@
But \emph{in which order} should we apply the introduced program
transformations?  In principle, one could try to blindly iterate the
proposed techniques and hope that a FOP can cope with the
output. Since transformations are closed under composition, the blind
iteration of transformations is sound, although seldom effective. In
short, a \emph{strategy} is required that combines the proposed
techniques in a sensible way. There is no clear notion of a
perfect strategy. After all, we are interested in non-trivial program
properties.  However, it is clear that any sensible strategy should at
least (i) yield overall a transformation that is
effectively computable, (ii)~not defeat its purpose by generating TRSs
whose runtime complexity is not at all in relation to the complexity
of the analysed program, and (iii)~produce ATRSs that FOPs are able
to analyse.

In Figure~\ref{fig:strategy} we render the current transformation
strategy underlying our tool \hoca.\@ More precise, 
Figure~\ref{fig:strategy} defines a transformation \strat|simplify| 
based on the following \emph{transformation combinators}:
\begin{varitemize}
\item 
  \strat|$\transone_1$;$\transone_2$| 
  denotes the composition $\transone_2 \circ \underline{\transone_1}$, where 
  $\underline{\transone_1}(\atrsone) = \transone_1(\atrsone)$ if defined and $\underline{\transone_1}(\atrsone) = \atrsone$ otherwise; 
\item 
  the transformation \strat|exhaustive$\transone$| iterates the
  transformation $\transone$ until inapplicable on the current problem; and 
\item 
  the operator \strat|<>| implements left-biased choice: 
  \strat|$\transone_1$ <> $\transone_2$| applies
  transformation $\transone_1$ if successful, otherwise $\transone_2$
  is applied.
\end{varitemize}
It is easy to see that all three combinators preserve the two crucial properties of transformations, 
viz, complexity reflection and complexity preservation. 

The transformation \strat|simplify| depicted in
Figure~\ref{fig:strategy} is composed out of three transformations
\strat|simpATRS|, \strat|toTRS| and \strat|simpTRS|, each itself
defined from transformations \strat|inline($P$)| and \strat|cfa|
describe in Sections~\ref{s:impl:inline} and~\ref{s:impl:cfa},
respectively, the transformation \strat|usableRules| which implements
the aforementioned computationally cheap, unification based, criterion
from~\cite{GTSK:FROCOS:05} to eliminate dead code (see
Section~\ref{s:simpl:de}), and the transformation \strat|uncurry|,
which implements the uncurrying-transformation from
Section~\ref{s:simpl:uncurrying}.

The first transformation in our chain, \strat|simpATRS|, performs
inlining driven by the specific shape of the input ATRS obtained by
defunctionalisation, followed by syntax driven dead code elimination.
The transformation \strat|toTRS| will then translate the intermediate
ATRSs to functional form by the uncurrying transformation, using
control flow analysis to instantiate head variables sufficiently and
further eliminate dead code.  The transformation \strat|simpTRS| then
simplifies the obtained TRS by controlled inlining, applying syntax
driven dead code elimination where possible, resorting to the more
expensive version based on control flow analysis in case the
simplification stales.  To understand the sequencing of
transformations in \strat|simpTRS|, observe that the strategy
\strat|inline(decreasing)| is interleaved with dead code
elimination.  Dead code elimination, both in the form of
\strat|usableRules| and \strat|cfaDCE|, potentially restricts the set
$\narrowings{p}{l \to r}$, and might facilitate in consequence the
transformation \strat|inline(decreasing)|.  Importantly, the rather
expensive, flow analysis driven, dead code analysis is only performed
in case both \strat|inline(decreasing)| and its cheaper cousin
\strat|usableRules| fail.

\shortlongv{
The overall strategy \strat|simplify| is well-defined 
on all inputs obtained by defunctionalisation, i.e., terminating~\cite{EV}.
}{
To see termination, it suffices to realize that all exhaustive applications of 
transformations in \strat|simplify| are terminating:
\begin{varitemize}
\item 
  For \strat|inline(match)| this claim is immediate by the shape of
  input ATRSs.  Each application of \strat|inline(match)| removes one
  occurrence of a closure-constructor obtained from the transformation
  of a match-expression in right-hand sides.
\item 
  Similar, exhaustive application of \strat|inline(constructor)| is
  terminating, since at each step the number of defined symbol in
  right-hand sides is reduced.
\item 
  For iterated application of \strat|inline(lambda-rewrite)| the claim
  is less obvious.  Intuitively, termination holds because the
  rewritings performed on right-hand sides correspond to steps with
  respect to a very restricted fragment of \PCF, which is itself
  terminating: the simply typed $\lambda$-calculus. Note that the
  restriction to rewrites is essential, as soon as we allow inlining
  by narrowing, termination is not guaranteed.
\item 
  Concerning the final case, by way of contradiction suppose that
  \[
  \strat{(inline(decreasing); usableRules) <> cfaDCE} \tkom
  \]
  is applied infinitely often. Dead code elimination cannot be the culprit, indeed, 
  \strat|inline(decreasing)| can then be applied infinitely often. 
  In such a sequence, the case \emph{proper inlining} underlying the definition 
  of the predicate \strat|decreasing| cannot hold infinitely often, as the 
  number of defined symbols in right-hand sides decrease after each application. Hence ultimately, 
  an infinite application of \strat|inline(decreasing)| has to happen due to the \emph{size decreasing} condition.
  But in such a sequence, the multiset of sizes of right-hand sides is decreasing with respect to the multiset 
  extension of the strict order $>$ on naturals, which itself is well-founded. Contradiction!
\end{varitemize}
} 
Although we cannot give precise bounds on the runtime complexity in
general, in practice the number of applications of inlinings is
sufficiently controlled to be of practical relevance.  Importantly,
the way inlining and instantiation is employed ensures that the sizes
of all intermediate TRSs are kept under tight control.


%% file: experiments.tex
\section{Experimental Evaluation}\label{s:experiments}

\newcommand{\lz}{\phantom{0}}
\newcommand{\defunct}{\textsf{D}}
\newcommand{\simplify}{\textsf{S}}
\newcommand{\pb}[1]{\,#1\,}
\newcommand\tm[3]{\pb{#1}/\pb{#2}/\pb{#3}}
\begin{table*}
  \shortv{\nocaptionrule}\caption{Experimental Evaluation conducted with \protect\tct\ and \protect\TTTT.}\label{t:experiments}
  \longv{\tiny}
  \centering
  \begin{tabular}{l r c c c c c}
    \toprule
 & \TOP\BOT\strut                 & constant              & linear                     & quadratic                  & polynomial                 & terminating           \\
    \midrule
    \TOP\defunct 
 & \strut\# systems               & 2                     & 5                          & 5                          & 5                          & 8                     \\
 & \BOT\strut{}FOP execution time & \tm{0.37}{1.71}{3.05} & \tm{0.37}{4.82}{13.85}    & \tm{0.37}{4.82}{13.85}    & \tm{0.37}{4.82}{13.85}    & \tm{0.83}{1.38}{1.87} \\
    \TOP
    \simplify
 & \strut\# systems               & 2                     & 14                         & 18                         & 20                         & 25                    \\
 & \strut\hoca\ execution time    & \tm{0.01}{2.28}{4.56} & \tm{0.01}{0.54}{\lz{}4.56} & \tm{0.01}{0.43}{\lz{}4.56} & \tm{0.01}{0.42}{\lz{}\,4.56} & \tm{0.01}{0.87}{6.48} \\
 & \BOT\strut{}FOP execution time & \tm{0.23}{0.51}{0.79} & \tm{0.23}{2.53}{14.00}     & \tm{0.23}{6.30}{30.12}     & \tm{0.23}{10.94}{60.10}    & \tm{0.72}{1.43}{3.43} \\
    \bottomrule
  \end{tabular}
\end{table*}

So far, we have covered the theoretical and implementation aspects
underlying our tool \hoca.\@ The purpose of this section is to
indicate how our methods performs in practice.
To this end, we compiled a diverse collection of higher-order
programs from the 
literature~\cite{DN:PPDP:01,JHLH:POPL:10,Okasaki:JFP:98} and standard 
textbooks~\cite{Bird,RabhiLapalme}, on which we performed tests with our tool in conjunction
with the general-purpose first-order complexity tool \tct~\cite{AM:RTA:13b},
version 2.1.\footnote{We ran also experiments with \aprove\ and \cat\ as
  back-end, this however did not extend the power.}
For comparison, we have also paired \hoca\ with the termination tool \TTTT~\cite{MKCSHZAM:RTA:09}, version 1.15.

In Table~\ref{t:experiments} we summarise our experimental findings on
the 25 examples from our collection.\footnote{Examples and full experimental evidence can be found on the \hoca\ homepage.}
Row \simplify\ in the table indicates the total number of
higher-order programs whose runtime could be classified linear,
quadratic and at most polynomial when \hoca\ is paired with the
back-end \tct, and those programs that can be shown terminating when
\hoca\ is paired with \TTTT.\@ In contrast, row \defunct\ shows 
the same statistics when the FOP is run directly on the defunctionalised program, 
given by Proposition~\ref{p:pcf2trs}.
To each of those results, we state the
minimum, average and maximum execution time of \hoca\ and the employed
FOP.\@ All experiments were conducted on a machine with a 8 dual core
AMD Opteron${}^\text{\texttrademark}$ 885 processors running at
2.60GHz, and 64Gb of RAM.\footnote{Average PassMark CPU Mark 2851;
\url{http://www.cpubenchmark.net/}.}\@ Furthermore, the
tools were advised to search for a certificate within 60
seconds.

As the table indicates, not all examples in the testbed are subject to
a runtime complexity analysis through the here proposed approach. However, 
at least termination can be automatically verified. 
For all but one example (namely \aut|mapplus.fp|)
the obtained complexity certificate is asymptotically optimal.
As far as we know, no other fully automatic complexity tool 
can handle the five open examples.
We will comment below on the reason why \hoca\ may fail.


Let us now analyse some of the programs from our testbed. For
each program, we will briefly discuss what \hoca, followed by selected
FOPs can prove about it. This will give us the opportunity to discuss
about specific aspects of our methodology, but also about limitations
of the current FOPs.
\paragraph{Reversing a List.} Our running example, namely the functional
program from Section~\ref{s:defunc} which reverses a list, can be
transformed by \hoca\ into an ATRS which can easily be proved to have
linear complexity. Similar results can be proved for other programs.
\paragraph{Parametric Insertion Sort.} A more complicated example is a
higher-order formulation of the insertion sort algorithm, example
\aut|isort-fold.fp|, which is parametric on the subroutine which
compares the elements of the list being sorted. This is an example
which cannot be handled by linear type systems~\cite{BT:IC:09}: we do recursion over a
function which in an higher-order variable occurs free. Also, type
systems like the ones in~\cite{JHLH:POPL:10}, which are restricted to
linear complexity certificates, cannot bind the runtime complexity of
this program. \hoca, instead, is able to put it in a form which allows
\tct\ to conclude that the complexity is, indeed quadratic.

\paragraph{Divide and Conquer Combinators.} Another noticeable example
is the divide an conquer combinator, defined in example \aut|mergesort-dc.fp|, which we have taken from
\cite{RabhiLapalme}. We have then instantiated it so that the resulting algorithm
is the merge sort algorithm. \hoca\ is indeed able to translate the
program into a first-order program which can then be proved to be
\emph{terminating} by FOPs. This already tells us that the obtained
ATRS is in a form suitable for the analysis. The fact that FOPs
cannot say anything about its complexity is due to the limitations of
current FOPS which, indeed, are not able to perform any non-local size
analysis, itself a necessary condition for proving merge sort to be a
polynomial time algorithm. Similar considerations hold for Okasaki's
parser combinator, various instances of which can be proved themselves 
terminating.


%% file: relatedwork.tex
\section{Related Work}\label{s:relatedwork}
What this paper shows is that complexity analysis of higher-order
functional programs can be made easier by way of program
transformations.  As such, it can be seen as a complement rather than
an alternative to existing methodologies. Since the literature on
related work
is quite vast, we will only
give in this section an overview of the state of the art, highlighting
the differences with to our work.

\paragraph{Control Flow Analysis.}
A clear understanding of control flow in higher-order programs is
crucial in almost any analysis of non-functional properties.
Consequently, the body of literature on control flow analysis is
considerable, see e.g.\ the recent survey
of~\citet{Midtgaard:ACMCS:12}.  Closest to our work, control flow
analysis has been successfully employed in termination analysis, for
brevity we mention only~\cite{PS97,JB:LMCS:08,GRSST:TOPLAS:11}.  By
\citet{JB:LMCS:08} a strict, higher-order language is studied, and
control flow analysis facilitates the construction of size-change
graphs needed in the analysis. Based on earlier work by \citet{PS97}, 
\citet{GRSST:TOPLAS:11} study termination of \haskell{} through so-called 
\emph{termination} or \emph{symbolic execution} graphs, 
which under the hood corresponds to a careful study of the control flow in \haskell{} programs.
While arguable \emph{weak dependency pairs}~\cite{HM:IJCAR:08} or 
\emph{dependency triples}~\cite{NEG:CADE:11} form a weak notion of
control flow analysis, our addition of collecting semantics to complexity
analysis is novel.

\paragraph{Type Systems.}
That the r\^ole of type systems can go beyond type safety is well-known.
The abstraction type systems implicitly provide,
can enforces properties like termination or bounded
complexity. In particular, type systems for the $\lambda$-calculus are
known which characterise relatively small classes of functions like
the one of polynomial time computable functions~\cite{BT:IC:09}. 
The principles underlying these type
systems, which by themselves cannot be taken as verification
methodologies, have been leveraged while defining type systems for
more concrete programming languages and type inference procedures, 
some of them being intensionally
complete~\cite{LG:LMCS:12,LP:POPL:13}. All these results are of course
very similar in spirit to what we propose in this work. What is
lacking in most of the proposed approaches is the presence, at the
same time, of higher-order, automation, and a reasonable expressive
power. As an example, even if in principle type systems coming from
light logics~\cite{BT:IC:09} indeed handle higher-order functions and
can be easily implementable, the class of catched programs is small
and full recursion is simply absent. 
On the other hand, \citet{JHLH:POPL:10} have successfully encapsulated 
Tarjan's amortised cost analysis into a type systems that allows 
a fully automatic resource analysis. 
In contrast to our work, only linear resource usage can be established. However,
their cost metric is general, while our technique only works
for time bounds. Also in the context of amortised analysis, 
\citet{Danielsson:2008} provides a semiformal verification
of the runtime complexity of lazy functional languages, which allows
the derivation of non-linear complexity bounds on selected examples.

\paragraph{Term Rewriting.}
Traditionally, a major concern in rewriting has been
the design of sound algorithmic methodologies for checking termination. 
This has given rise to many different techniques including
basic techniques like path orders or interpretations, as well
as sophisticated transformation techniques, c.f.~\cite[Chapter~6]{Terese}.
Complexity analysis of TRSs can be seen as a natural generalisation
of termination analysis. And, indeed, variations on path orders and
the interpretation methods capable of guaranteeing quantitative 
properties have appeared one after the other starting from the beginning
of the nineties~\cite{BCMT:JFP:2001,Marion:IC:03,AM:LMCS:13}.
In both termination and complexity analysis, the rewriting community
has always put a strong emphasis to automation. However, with respect
to higher-order rewrite systems (HRSs) only termination has received steady
attention, c.f.~\cite[Chapter~11]{Terese}. Except for very few attempts without 
any formal results complexity analysis of HRSs has been 
lacking~\cite{BMP:ICTAC:07,BL:CSL:12}.

\paragraph{Cost Functions.}
An alternative strategy for complexity analysis consists in
translating programs into other expressions (which could be programs
themselves) whose purpose is precisely computing the complexity (also
called the \emph{cost}) of the original program. Complexity analysis
is this way reduced to purely extensional reasoning on the obtained
expressions. Many works have investigated this approach in the context
of higher-order functional languages, starting from the pioneering
work by Sands~\cite{Sands:ESOP:90} down to more recent contributions,
e.g. \cite{VH:IFL:03}. What is common among most of the cited works
is that either automation is not considered (e.g. cost functions can
indeed be produced, but the problem of putting them in closed form is
not~\cite{VH:IFL:03}), or the time complexity is not analysed
parametrically on the size of the input~\cite{GL:PEPM:02}. A notable
exception is Benzinger's work~\cite{Benzinger:TCS:04}, which however
only applies to programs extracted from proofs, and thus only works
with primitive recursive definitions.
